\documentclass[12pt]{elsarticle}

\usepackage[T1]{fontenc}
\usepackage{color}
\usepackage{array}
\usepackage{colortbl}
\usepackage{nomencl}
\makenomenclature
\usepackage{amssymb}
\usepackage{amsmath}
\usepackage{algorithm}
\usepackage{algorithmicx,algpseudocode}
\usepackage{etoolbox}
\usepackage{mathtools}
\usepackage{epsfig}
\usepackage{tabularx}
\usepackage{multirow}
\usepackage{subfigure}
\usepackage{babel}
\usepackage{psfrag}
\usepackage{paralist}
\usepackage{setspace}
\usepackage{comment}
\usepackage{morefloats}
\usepackage{vale}
\usepackage{enumitem}
\usepackage[normalem]{ulem}
\usepackage[hidelinks,pdfencoding=auto, psdextra]{hyperref}
\hypersetup{pdfauthor={Name}}
\usepackage{bm}
\usepackage[a4paper, margin=1in]{geometry}
\PassOptionsToPackage{pdftex}{graphicx}
\usepackage{tikz}
\usepackage{algorithmicx,algpseudocode}
\usepackage[compact]{titlesec}
\usepackage{titlesec}
\usepackage{url}
\usepackage{enumitem}
\usepackage{amsthm}

\usepackage{listings}
\usepackage{xcolor}

\lstdefinestyle{cppcompact}{
  language=C++,
  basicstyle=\ttfamily\small,
  numbers=left,
  numberstyle=\tiny,
  numbersep=6pt,
  frame=single,
  breaklines=true,
  columns=fullflexible,
  keepspaces=true,
  showstringspaces=false,
  tabsize=2
}

\newtheorem{theorem}{Theorem}[section]

\newtheorem{proposition}[theorem]{Proposition}
\newtheorem{definition}[theorem]{Definition}

\newcounter{bla}

\newcommand{\statevec}[1]{\boldsymbol{\mathsf{#1}}}

\journal{Journal of Computational Physics}

\begin{document}

\begin{frontmatter}

\title{On the transmission of floating-point perturbations in flow--dependent filter--width formulations in\\ Large--Eddy Simulation}

\author[ancona]{Valerio D'Alessandro \corref{cor}}
\ead{v.dalessandro@univpm.it}
\cortext[cor]{Corresponding author}
\author[cineca_rm]{Alessio Piccolo}
\author[ancona]{Matteo Falone}
\author[cineca_bo]{Simone Bnà}


\address[ancona]{Dipartimento di Ingegneria Industriale e Scienze Matematiche,\\ Universit\`a Politecnica delle Marche, Via Brecce Bianche 12, 60131, Ancona, Italy}
\address[cineca_rm]{CINECA, Direzione HPC, Piazza dell'Indipendenza 11/B, 00185, Rome, Italy}
\address[cineca_bo]{CINECA, Direzione HPC, Via Magnanelli 2, 40033, Casalecchio di Reno (BO), Italy}

\begin{abstract}
Heterogeneous high--performance computing architectures expose numerical algorithms
to perturbations arising from the non--associativity of floating--point arithmetic. In Large--Eddy Simulation (LES), similar numerical effects may become relevant when they affect
the filter--width entering the subgrid scale (SGS) model. This work investigates this mechanism for the least--squares (LSQ) based filter--width formulation, focusing on how floating--point effects are generated, transmitted, and coupled with the resolved flow. We show that, for fixed resolved kinematics, the LSQ
filter--width is logarithmically non-expansive but not strictly contractive with respect
to perturbations of the mesh metrics. Consequently, small disturbances may be transmitted with little attenuation through strongly directional filter-width responses.  To mitigate this sensitivity, we introduce a scalar max--min compression of the directional mesh scales together with a bounded modulation based on the resolved velocity gradient. The resulting formulation reroutes floating--point perturbations
through the filter-width operator, improving robustness while preserving the
flow--dependent character of the original LSQ construction.\\
The framework is assessed on heterogeneous CPU and GPU architectures for flow past
a circular cylinder at $\mathrm{Re}=3900$ and the Taylor--Green vortex at $\mathrm{Re}=1600$. 
In the former, nearly one--to--one transmission of relative metric disturbances can become relevant when the transmitted perturbations interact with shear-layer transition.
By contrast, on orthogonal Taylor--Green vortex grids, the accumulation
pathway is structurally absent and no comparable macroscopic response develops.
These results suggest floating--point sensitivity matters for LES filter--width formulations in heterogeneous computing environments.
\end{abstract}

\end{frontmatter}

\section{Introduction} \label{sec:intro}
Since its introduction, Large--Eddy Simulation (LES) has become an established approach in fluid dynamics, bridging the gap between the high computational cost of direct numerical simulation and the modeling limitations of Reynolds-Averaged Navier-Stokes (RANS) approaches.
One of the main barriers to the widespread industrial adoption of LES has traditionally been its computational cost~\cite{ChoiMoin2012,BosePark2018}. However, the rapid growth of heterogeneous high--performance computing (HPC) systems (based on the integration of CPU and GPU) has substantially alleviated this limitation, enabling simulations at unprecedented spatial and temporal resolutions,~\cite{Roccon2026}. At the same time, massively parallel architectures introduce numerical effects that are less prominent in conventional CPU based implementations. In particular, owing to the non--associativity of floating-point arithmetic, parallel accumulation operations may introduce small round–-off perturbations that, in turn, can propagate through LES filtering and subgrid--scale (SGS) modeling operations.\\
%
%
Discretizations and modeling errors in LES have been extensively investigated,~\cite{Ghosal1996,Sagaut2006}, whereas the numerical implications related to modern heterogeneous hardware remain fairly underexplored.
This issue is particularly relevant because filtering and SGS modeling introduce additional mathematical operators through which small numerical perturbations may enter and, eventually, alter the resolved--scale dynamics.
In addition, recent papers indicate that specific numerical formulations and derived quantities may retain a marked sensitivity to finite--precision arithmetic perturbations in LES,~\cite{KARP2026114600}. Nevertheless, turbulence statistics can remain remarkably robust under reduced floating--point precision,~\cite{KARP2026114600}.
%
%
Against this background, the present work focuses on flow--dependent filter--width formulations.
These approaches are particularly relevant on anisotropic grids, where the definition of an appropriate SGS length scale has been recognized as a non--trivial problem. For example, early attempts to explicitly account for grid anisotropy were proposed by Scotti et al.~\cite{Scotti1993}. 
However, more recent approaches extend this idea by incorporating not only the local grid geometry but also the orientation of the resolved flow.
In particular, Trias et al.,~\cite{trias2017,trias:2025}, introduced the least--squares length (LSQ), which has proved effective for LES on anisotropic grids in both incompressible and compressible flow problems,~\cite{trias2017,dalessandro2024}.
Notably, its dependence on the directional grid metrics provides a mechanism through which small perturbations of these quantities can be transmitted to the SGS closure.\\
In this paper, the aforementioned behavior is investigated from both mathematical and computational perspectives.
First, the response of the standard LSQ operator to perturbations of the directional mesh metrics is characterized for fixed resolved kinematics.
The analysis shows that the LSQ mapping is logarithmically non expansive but not strictly contractive, so that metric perturbations may be transmitted with little attenuation when the resolved kinematics strongly favor one directional contribution.
%
To control this mechanism, we introduce a grid anisotropy aware scalar scale obtained through a max--min compression of the directional mesh metrics. The resulting filter reroutes metric perturbation transmission while preserving a bounded dependence on the resolved kinematics.
Additionally, the theoretical analysis is complemented by numerical experiments on state--of--the--art heterogeneous CPU and GPU architectures, considering two flow cases: the flow past a circular cylinder at $\mathrm{Re}=3900$ and the Taylor--Green vortex (TGV) at $\mathrm{Re}=1600$.
Several diagnostic strategies are also devised and employed to isolate the respective roles of metric assembly and the flow--dependent part of the LSQ formulation.
%
%
%
%
The cylinder flow shows that the nearly one--to--one transmission of  perturbations through the filter--width operator can become relevant when coupled with a receptive transitional shear--layer. On the other hand, the TGV, on the adopted orthogonal grid, provides a case in which the corresponding directional accumulation mechanism is structurally absent.\\
Finally, all the proposed developments are implemented within SPUMA~\cite{bna2025spuma}, a minimally invasive GPU 
oriented framework for OpenFOAM. It enables the execution of LES solvers on modern 
accelerator-based HPC systems while preserving the original algorithmic structure 
of the underlying OpenFOAM implementation.\\
The remainder of this paper is organized as follows.  After introducing the governing equations, SGS modeling, and filter--width formulations in Sec.~\ref{sec:gov}, Sec.~\ref{sec:theory} examines the transmission of metric perturbations and the structural properties of the proposed formulation. The corresponding implementation strategies are then discussed in Sec.~\ref{sec:impl}, before the numerical results are presented in Sec.~\ref{sec:results}. The conclusions are finally summarized in Sec.~\ref{sec:conclusions}. 
\section{Governing equations}\label{sec:gov}
The governing equations are the Favre--filtered compressible Navier--Stokes equations, written in conservative form as
\begin{equation}
\frac{\partial \bar{\uvec}}{\partial t}
+ \NABLA \cdot {\statevec F}_{\mathrm{c}}(\bar{\uvec})
- \NABLA \cdot {\statevec F}_{\mathrm{v}}(\bar{\uvec},\NABLA\bar{\uvec})
+ \NABLA \cdot {\statevec F}_{\mathrm{SGS}}(\bar{\uvec},\NABLA\bar{\uvec})
=  {\bf 0}
\qquad
\text{in } Q~,
\label{eq:ns-les}
\end{equation}
where $Q=\Omega\times[0,T]$, $\Omega\subset\mathbb{R}^3$ denotes the computational domain and $T$ is the final simulation time.
The system is supplemented with suitable initial and boundary conditions
\begin{equation}
\begin{aligned}
\bar{\uvec} &= \bar{\uvec}_D                              && \text{on } \Gamma_D \times [0,T]~, \\
\frac{\partial \bar{\uvec}}{\partial x_j}n_j &= {\bf h}_N && \text{on } \Gamma_N \times [0,T]~, \\
\bar{\uvec}(\mathbf{x},0) &= \bar{\uvec}_0(\mathbf{x})    && \text{in } \Omega~,
\end{aligned}
\end{equation}
where $\Gamma=\Gamma_D\cup\Gamma_N$ is the boundary of the computational domain, $\Gamma_D$ and $\Gamma_N$ refer to the Dirichlet and Neumann boundaries, respectively, and $n_j$ are the components of the outward unit normal vector.\\
Relying on the vector $ \bar{\uvec} = \left( \bar{\rho}, \bar{\rho}\tilde{u}_i,
\bar{\rho}\tilde{E} \right)^T$, the j-th components of the convective, diffusive, and subgrid-scale fluxes read
\begin{equation}
{\statevec f}_{\mathrm{c,j}}
=
\begin{pmatrix}
\bar{\rho}\tilde{u}_j\\
\bar{\rho}\tilde{u}_1\tilde{u}_j+\bar{p}\delta_{1j}\\
\bar{\rho}\tilde{u}_2\tilde{u}_j+\bar{p}\delta_{2j}\\
\bar{\rho}\tilde{u}_3\tilde{u}_j+\bar{p}\delta_{3j}\\
\bar{\rho}\tilde{u}_j\tilde{H}
\end{pmatrix}~,
\qquad
{\statevec f}_{\mathrm{v,j}}
=
\begin{pmatrix}
0\\
\bar{\sigma}_{1j}\\
\bar{\sigma}_{2j}\\
\bar{\sigma}_{3j}\\
\bar{\sigma}_{ji}\tilde{u}_i-\tilde{q}_j
\end{pmatrix}~,
\qquad
{\statevec f}_{\mathrm{SGS,j}}
=
\begin{pmatrix}
0\\
\tau_{1j}\\
\tau_{2j}\\
\tau_{3j}\\
\phi
\end{pmatrix}~.
\label{eq:fsgs}
\end{equation}
%
%
%
The viscous stress tensor and heat flux vector are defined according to the constitutive relation for Newtonian fluids and Fourier's law:
%
\begin{equation}\label{eq:viscous-constitutive}
		{\tilde{\sigma}_{ij}} = 2\mu \left( \tilde{S}_{ij} - \frac{1}{3}  \frac{ \partial \tilde{u}_k}{\partial x_k} \delta_{ij} \right)~,  \quad
		 \tilde{q}_j = -\lambda \frac{\partial \tilde{T}}{\partial x_j}~.\\
\end{equation}
%
In the above equations, $\mu$ is the molecular viscosity and $\tilde{S}_{ij} = \frac{1}{2} \left( \frac{\partial \tilde{u}_i}{\partial x_j}
+
\frac{\partial \tilde{u}_j}{\partial x_i}
\right)
$
is the Favre--filtered strain-rate tensor. The thermal conductivity, $\lambda$, is obtained from the Prandtl number, $\mathrm{Pr}=\mu c_p/\lambda$, where $c_p$ is the specific heat at constant pressure.
The pressure is computed from the ideal gas equation of state, while the temperature is evaluated from the total energy, neglecting the SGS kinetic energy contribution at the low Mach numbers considered here:
%
\begin{equation}\label{eq:pT}
		 \bar{p} = \bar{\rho} \left( \gamma - 1 \right) \left( \tilde{E} - \frac{1}{2}  \tilde{u}_k \tilde{u}_k \right), \quad
		 c_v \tilde{T} = \tilde{E} - \frac{1}{2} \tilde{u}_k \tilde{u}_k~, 
\end{equation}
%
with $\gamma$ heat capacity ratio  and $c_v$ specific heat at constant volume. The anisotropic part of the SGS stresses is parametrized using the eddy--viscosity hypothesis:
\begin{equation}\label{eq:fc}
		{\tau}_{ij} - \frac{1}{3}\tau_{kk} \delta_{ij} = -2\mu_{\mathrm{SGS}} \left( \tilde{S}_{ij} - \frac{1}{3}  \frac{ \partial \tilde{u}_k}{\partial x_k} \delta_{ij} \right)
\end{equation}
where the SGS dynamic viscosity is evaluated as
\begin{equation}\label{eq:bous}
\mu_{\mathrm{SGS}}
=
\bar{\rho} C_s^2 \Delta^2
|\widetilde{\mathbf{S}}|~,
\end{equation}
Here, $|\widetilde{\mathbf{S}}|$
corresponds to the magnitude of the Favre--filtered strain--rate tensor and $\Delta$ is the spatial filter width. The Smagorinsky coefficient, $C_s^2$, is dynamically computed using the Germano--Lilly procedure~\cite{GermanoLilly}, whereas the SGS contribution to the energy equation, $\phi$, is modeled following Pino Mart\'in et al.~\cite{pinomartin2000}.
On the other hand, the terms involving SGS 
turbulent diffusion and SGS viscous diffusion are neglected, as they 
produce a negligible effect at low Mach 
numbers~\cite{pinomartin2000, garnier2009}.\\
To prevent spurious acoustic wave reflections at the far--field boundaries, a sponge--layer technique is employed. 
Hence, eq.~\ref{eq:ns-les} is rewritten as:
\begin{equation}
\frac{\partial \bar{\uvec}}{\partial t}
+ \NABLA \cdot {\statevec F}_c(\bar{\uvec})
- \NABLA \cdot {\statevec F}_v(\bar{\uvec},\NABLA\bar{\uvec})
+ \NABLA \cdot {\statevec F}_{\mathrm{SGS}}(\bar{\uvec},\NABLA\bar{\uvec})
=  
\sigma \left( \bar{\uvec}_{ref} - \bar{\uvec} \right)
\qquad
\text{in } Q~;
\label{eq:ns-les1}
\end{equation}
the unphysical source term added to right--hand side is intended to damp the resolved quantities to a user defined reference value, $\bar{\uvec}_{ref}$, typically selected equal to the undisturbed free--stream conditions.
The damping coefficient $\sigma$ is defined as
%
\begin{equation}\label{eq:sponge}
\sigma = \sigma_0 \left( \frac{L_{sp} - d}{L_{sp}}  \right)^n
\end{equation}
%
where $L_{sp}$ is the sponge--layer width, $d$ corresponds to the minimum distance from the outer boundary, $\sigma_0$ controls the damping strength, and $n$ determines the shape of the sponge profile. In the present work, $n$ is fixed to $2$, while the sponge strength and layer width are selected according to the procedure described by D'Alessandro et al.~\cite{dalessandro2020caafoam}.
\subsection{Baseline filter--width}\label{sec:filter}
\noindent As a baseline, we consider the least--squares (LSQ) 
filter width proposed by Trias et al.~\cite{trias2017}:
\begin{equation}
\Delta_{\mathrm{LSQ}} = \sqrt{\frac{\boldsymbol{\Delta}_D \mathbf{G}^T \mathbf{G} : 
\boldsymbol{\Delta}_D \mathbf{G}^T \mathbf{G}}{\mathbf{G}^T \mathbf{G} : 
\mathbf{G}^T \mathbf{G}}}
\label{eq:delta-lsg}
\end{equation}
where $\mathbf{G}$ is the velocity gradient tensor, whose components are obtained from Green's theorem:
\begin{equation}{\label{eq:les7}}
    \left( \frac{\partial  \Tilde{u}_i}{\partial x_j} \right)_P = \frac{1}{2V_P} 
\sum_{N \in \mathcal{N}\left(P\right)} \left[  \left(\Tilde{ u}_i\right)_P +  \left(\Tilde{u}_i\right)_N \right] S_{PN} 
\left( n_{x_j} \right)_{PN}
 \end{equation}
with $\left( n_{x_j} \right)_{PN} = \left[ n_x, n_y, n_z\right]_{PN}^T$ denoting the outward--pointing
unit normal vector to the interface $PN$ shared by cells $P$ and $N$ and $V_p$ is the cell's volume. 
The quantity $S_{PN}$ is the face area of the interface $PN$, while $\mathcal{N}(P)$ is the set of direct neighbours of cell $P$.
Conversely, 
$\boldsymbol{\Delta}_D=\mathrm{diag}(\Delta_x,\Delta_y,\Delta_z)$
is the tensor collecting the directional grid length scales
computed with the discrete gradient operator:
\begin{equation}
(\Delta_{x_i})_P = \frac{2V_P}{\sum_{N \in \mathcal{N}(P)} S_{PN} |n_{x_i}|_{PN}}~.
\label{eq:deltaxj}
\end{equation}
%
%

\section{Metric perturbation transmission and regularization of flow--dependent filter--widths} \label{sec:theory}
This section examines how perturbations introduced during the assembly of the directional mesh metrics are transmitted through the LSQ based filter--width formulation.
A regularization to prevent singularities in the filter--width and a specific strategy to control the transmission of round--off disturbances are also introduced and discussed.
\subsection{Metric perturbation transmission through the standard LSQ filter formulation}
\label{sec:lsq-sensitivity}
%
%
\noindent Defining $\mathbf{A} = \mathbf{G}^{T}\mathbf{G}$
eq.~\ref{eq:delta-lsg} can be written, for $\mathbf{A}:\mathbf{A}>0$, as
\begin{equation}
\label{eq:lsq-A-form}
\Delta_{\mathrm{LSQ}} = 
\sqrt{
\frac{
\left(
\boldsymbol{\Delta}_{D}\mathbf{A}
\right)
:
\left(
\boldsymbol{\Delta}_{D}\mathbf{A}
\right)
}{
\mathbf{A}:\mathbf{A}
}
}~.
\end{equation}
Since $\boldsymbol{\Delta}_{D}$ is diagonal,
\begin{equation}
\left(
\boldsymbol{\Delta}_{D}\mathbf{A}
\right)
:
\left(
\boldsymbol{\Delta}_{D}\mathbf{A}
\right)
=
\sum_{i,j}
\Delta_{x_i}^{\,2} A_{ij}^{\,2}
=
\sum_i
\Delta_{x_i}^{\,2}
\sum_j A_{ij}^{\,2}
\end{equation}
eq.~\ref{eq:lsq-A-form} admits the equivalent representation
\begin{equation}
\label{eq:lsq-weighted-mean}
\Delta_{\mathrm{LSQ}}^{2}
=
\sum_{i}
w_i(\mathbf{G})\Delta_{x_i}^{2},
\end{equation}
where
\begin{equation}
\label{eq:lsq-kinematic-weights}
w_i(\mathbf{G})
=
\frac{
\displaystyle\sum_{j} A_{ij}^{2}
}{
\mathbf{A}:\mathbf{A}
}~,
\qquad
w_i\geq0~,
\qquad
\sum_{i} w_i=1~.
\end{equation}
Consequently, the LSQ filter is a weighted quadratic mean of the directional grid scales, $\Delta_{x_i}$, with weights determined by the local resolved kinematics. The directional scales depend on geometric quantities assembled through finite--precision face accumulations and therefore inherit the corresponding perturbations. To isolate their transmission through the LSQ operator, the velocity gradient tensor $\mathbf{G}$ is held fixed. For a generic cell $P$, eq.~\ref{eq:deltaxj} can be written as
\begin{equation}
\label{eq:directional-scale-accumulator}
\Delta_{x_i}
=
\frac{2V_P}{\mathcal{A}_i}~,
\qquad
\mathcal{A}_i
=
{\sum_{N \in \mathcal{N}(P)} S_{PN} |n_{x_i}|_{PN}}~.
\end{equation}
Here, $\mathcal{A}_i$ denotes the geometric quantity evaluated through the parallel accumulation. For this reason, we write its perturbed counterpart as $\widehat{\mathcal{A}}_i=\mathcal{A}_i(1+\varepsilon_i)$, with $\varepsilon_i$ denoting the corresponding relative accumulation perturbation.
Assuming that the cell volume is unaffected by the accumulation order, the
perturbed directional scale satisfies
\begin{equation}
\label{eq:directional-scale-ratio}
r_i = \frac{\widehat{\Delta}_{x_i}}{\Delta_{x_i}}
=
\frac{1}{1+\varepsilon_i}~.
\end{equation}
$\widehat{\Delta}_{x_i}$ must remain positive and well defined, requiring $1+\varepsilon_i>0$~. Since $\varepsilon_i$ represents a round--off perturbation arising from floating--point accumulation, $|\varepsilon_i|\ll1$ in the regime of interest, and the above condition is naturally satisfied.\\
Keeping $\mathbf{G}$ fixed,  the LSQ filter evaluated using the perturbed
directional scales is
\begin{equation}
\label{eq:perturbed-lsq-r}
\widehat{\Delta}_{\mathrm{LSQ}}^{\,2}
=
\sum_i
w_i(\mathbf{G})
\Delta_{x_i}^{\,2}
r_i^{2}~.
\end{equation}
Dividing eq.~\ref{eq:perturbed-lsq-r} by
eq.~\ref{eq:lsq-weighted-mean} gives
\begin{equation}
\label{eq:lsq-ratio-convex}
\left(
\frac{
\widehat{\Delta}_{\mathrm{LSQ}}
}{
\Delta_{\mathrm{LSQ}}
}
\right)^2
=
\sum_i
\alpha_i r_i^{2}~,
\end{equation}
where
\begin{equation}
\label{eq:lsq-sensitivity-weights}
\alpha_i
=
\frac{
w_i(\mathbf{G})\Delta_{x_i}^{\,2}
}{
\displaystyle
\sum_k
w_k(\mathbf{G})\Delta_{x_k}^{\,2}
}~,
\qquad
\alpha_i\geq 0~,
\qquad
\sum_i\alpha_i=1~.
\end{equation}
The right-hand side of eq.~\ref{eq:lsq-ratio-convex} is a convex
combination of the $r_i^2$. Therefore,
\begin{equation}
\min_i r_i^{2}
\leq
\sum_i\alpha_i r_i^{2}
\leq
\max_i r_i^{2}~
\Longrightarrow
\min_i r_i
\leq
\frac{
\widehat{\Delta}_{\mathrm{LSQ}}
}{
\Delta_{\mathrm{LSQ}}
}
\leq
\max_i r_i~.
\end{equation}
%
A corresponding bound in terms of the ordinary relative perturbation follows
immediately:
\begin{equation}
\label{eq:lsq-relative-perturbation-bound}
\frac{
\left|
\widehat{\Delta}_{\mathrm{LSQ}}
-
\Delta_{\mathrm{LSQ}}
\right|
}{
\Delta_{\mathrm{LSQ}}
}
\leq
\max_i
\frac{
\left|
\widehat{\Delta}_{x_i}
-
\Delta_{x_i}
\right|
}{
\Delta_{x_i}
}
=
\max_i
\frac{\left| \varepsilon_i\right|}{1+\varepsilon_i}~.
\end{equation}
Therefore, eq.~\ref{eq:lsq-relative-perturbation-bound} provides a global
finite--perturbation bound: for fixed resolved kinematics, the relative
perturbation of the LSQ filter cannot exceed the largest relative perturbation
of the directional scales. This establishes the non--expansive character of
the LSQ mapping. However, the bound does not imply strict contraction, since
it does not provide a uniform attenuation factor smaller than unity for the quantity on the left-hand side of eq. \ref{eq:lsq-ratio-convex}.\\
The local response of the LSQ filter can be further characterized by
introducing the auxiliary relative perturbation
\begin{equation}
\label{eq:xi-definition}
\xi_i
:=
r_i-1
=
\frac{
\widehat{\Delta}_{x_i}-\Delta_{x_i}
}{
\Delta_{x_i}
}~;
\end{equation}
hence, $r_i=1+\xi_i$~,
and the unperturbed configuration corresponds to
$\boldsymbol{\xi}=\mathbf{0}$. From
eq.~\ref{eq:lsq-ratio-convex}, the exact logarithmic response of the
LSQ filter can be written as
\begin{equation}
\label{eq:lsq-log-response-xi}
\log
\left(
\frac{
\widehat{\Delta}_{\mathrm{LSQ}}
}{
\Delta_{\mathrm{LSQ}}
}
\right)
=
\frac{1}{2}
\log \left[ \sum_i \alpha_i \left( 1+\xi_i \right)^2\right]~.
\end{equation}
%
For $\|\boldsymbol{\xi}\|_{\infty}\ll 1$, a Taylor expansion of
eq.~\ref{eq:lsq-log-response-xi} gives
\begin{equation}
\label{eq:lsq-xi-first-order}
\log
\left(
\frac{
\widehat{\Delta}_{\mathrm{LSQ}}
}{
\Delta_{\mathrm{LSQ}}
}
\right)
=
\sum_i
\alpha_i\xi_i
+
\mathcal{O}
\left(
\|\boldsymbol{\xi}\|_{\infty}^{2}
\right)~.
\end{equation}
%
Accordingly, at first order, the logarithmic response of the LSQ filter is a
weighted average of the directional perturbations. The local transmission
depends on both the sensitivity weights $\alpha_i$ and the directional
coherence of the perturbations. In particular, if one directional contribution
dominates, $\alpha_k \simeq 1$, then
\begin{equation}
    \log\left(
    \frac{\widehat{\Delta}_{\mathrm{LSQ}}}
         {\Delta_{\mathrm{LSQ}}}
    \right)
    \simeq \xi_k~.
    \label{eq:one_to_one}
\end{equation}
%
The unit coefficient on the right-hand side of eq.~\ref{eq:one_to_one}
corresponds to the limiting non--expansive response: the perturbation is
transmitted without first-order amplification or attenuation.
Differently, a strictly contractive mapping would instead require a
local gain uniformly bounded by some $\kappa<1$, namely
\begin{equation}
\left|
\log\left(
\frac{\widehat{\Delta}_{\mathrm{LSQ}}}
     {\Delta_{\mathrm{LSQ}}}
\right)
\right|
\leq
\kappa |\xi_k|~,
\qquad \kappa<1~,
\end{equation}
so that the perturbation of the active directional scale is strictly
attenuated. 
It is important to note that eq.~\ref{eq:one_to_one} shows that the LSQ mapping does not admit a uniform contraction factor strictly smaller than unity, since its local gain can approach
unity when the weighting becomes strongly localized.
In terms of the ordinary relative perturbation, eq.~\ref{eq:one_to_one} can be also written as
\begin{equation}
\frac{\widehat{\Delta}_{\mathrm{LSQ}}-\Delta_{\mathrm{LSQ}}}
     {\Delta_{\mathrm{LSQ}}}
\simeq \xi_k~,
\end{equation}
showing that the relative perturbation of the dominant directional scale is
preserved at first order. Therefore, when the kinematic weighting becomes
strongly localized, the LSQ mapping approaches the limiting non--expansive
response, with a local transmission gain approaching unity. Hence, for fixed
$\mathbf{G}$, the first--order LSQ mapping is non--expansive but not strictly
contractive.
%
%

\subsubsection{Kinematic regularization of the LSQ filter}
\label{sec:lsq-kinematic-regularization}
\noindent To examine the role of the kinematic weighting independently of the
metric-assembly procedure, we introduce a Tikhonov inspired isotropic
diagonal loading of the velocity gradient tensor:
\begin{equation}
    \mathbf{G}_c
    =
    \mathbf{G}
    +
    \lambda_{\eta}\mathbf{I}~,
    \qquad
    \lambda_{\eta}
    =
    \eta \lVert\mathbf{G}\rVert_F~,
\label{eq:G_tikhonov}
\end{equation}
where $\eta$ is a dimensionless regularization parameter. The corresponding
kinematic tensor is
\begin{equation}
    \mathbf{A}_c
    =
    \mathbf{G}_c^{T}\mathbf{G}_c
    =
    \mathbf{A}
    +
    \lambda_{\eta}
    \left(\mathbf{G}^{T}+\mathbf{G}\right)
    +
    \lambda_{\eta}^{2}\mathbf{I}~.
\label{eq:A_tikhonov}
\end{equation}
When the filter appearing in eq.~\ref{eq:lsq-weighted-mean} is evaluated using $\mathbf{A}_c$ instead of $\mathbf{A}$, the
kinematic weights become
\begin{equation}
    w_i(\mathbf{G}_c)
    =
    \frac{
    \displaystyle\sum_j (A_c)_{ij}^{2}
    }{
    \mathbf{A}_c:\mathbf{A}_c
    }~,
\label{eq:w_tikhonov}
\end{equation}
and the standard LSQ formulation is recovered as $\eta\rightarrow0$.\\
The diagonal loading leaves the directional mesh scale assembly unchanged and acts exclusively on the kinematic weights entering the LSQ operator.
It redistributes the LSQ weights toward kinematically weak directions, thereby reducing, if present, the strong directional concentration responsible for near one--to--one metric perturbation transmission.
Note that the loading is used only as a diagnostic perturbation, with $\eta$ kept sufficiently small to preserve the dominant resolved kinematic structure. Larger values would instead amount to a modification of the filter-width model.
\subsection{Logarithmic stability and branchwise sensitivity of the max--min metric mapping}
\label{sec:sub-maxmim}
\noindent In order to control metric perturbation transmission from concentrated kinematic weighting, the directional mesh scales are compressed through a scalar max--min mapping:
\begin{equation}
\label{eq:maxmin-map-definition}
F(\Delta_x,\Delta_y,\Delta_z)
=
\max
\left[
\left(
\Delta_x\Delta_y\Delta_z
\right)^{1/3},
C_\Delta
\min
\left(
\Delta_x,\Delta_y,\Delta_z
\right)
\right].
\end{equation}
%
%
As shown in Appendix~A, for any positive directional scales $\Delta_{x_i}$ and corresponding perturbed values $\widehat{\Delta}_{x_i}$, the mapping in eq.~\ref{eq:maxmin-map-definition} satisfies
\begin{equation}
\label{eq:maxmin-log-nonexpansive-main}
\left|
\log
\left[
\frac{
F(\widehat{\Delta}_x,\widehat{\Delta}_y,\widehat{\Delta}_z)
}{
F(\Delta_x,\Delta_y,\Delta_z)
}
\right]
\right|
\leq
\max_{i\in{x,y,z}}
\left|
\log
\left(
\frac{\widehat{\Delta}_{i}}{\Delta_{i}}
\right)
\right|~.
\end{equation}
Hence, the complete mapping is globally non-expansive in the componentwise
logarithmic metric.
However, the mapping consists of two contributions: the geometric term,
$ \Delta_g = \left(\Delta_x\Delta_y\Delta_z\right)^{1/3}$, and the minimum--scale term,
$C_\Delta \Delta_{\min}$. They exhibit different local sensitivities.
When the geometric branch is active, ($\Delta=\Delta_g$), at first order
\begin{equation}
\frac{\delta \Delta_g}{\Delta_g}
\simeq
\frac{1}{3}
\left(
\frac{\delta \Delta_x}{\Delta_x}
+
\frac{\delta \Delta_y}{\Delta_y}
+
\frac{\delta \Delta_z}{\Delta_z}
\right)~.
\end{equation}
Thus, the geometric branch combines the directional perturbations through their arithmetic mean in logarithmic variables. A perturbation affecting only one directional scale is therefore reduced by a factor of $1/3$, whereas perturbations affecting all three directions are transmitted through their mean.\\ 
In the lower--bound branch, ~\emph{i.e.} $\Delta=C_\Delta\Delta_{\min}$, and provided that the direction attaining $\Delta_{\min}$ remains unchanged, its logarithmic variation is $\delta\log\Delta  = \delta\log\Delta_{\min}$.  At first order:
\begin{equation}
\frac{\delta\Delta}{\Delta}
\simeq
\frac{\delta\Delta_{\min}}{\Delta_{\min}}~.
\end{equation}
Consequently, the lower--bound branch transmits the relative perturbation of the active minimum scale one--to--one, while the complete max--min mapping remains globally non--expansive.\\
Finally, we want to remark that the coefficient $C_\Delta$  in eq.~\ref{eq:maxmin-map-definition} sets the
magnitude of the lower bound but does not alter the local relative
sensitivity within this branch. Its role is instead to control the
switching condition, and hence the regions in which the geometric or
lower-bound response is active. 
%
%
%
%
%
\subsection{LSQ filter formulation with compressed spatial metric and bounded kinematic modulation}
\label{sec:kinematic_regularization}
\noindent Under isotropic metric compression, the directional metric tensor entering
the standard LSQ construction is replaced by
$\boldsymbol{\Delta}_D=\Delta\mathbf{I}$ with $\Delta$ obtained from eq.~\ref{eq:maxmin-map-definition}. This modification provides the
perturbation control properties discussed above, but it has an immediate
structural consequence deriving from eq.~\ref{eq:lsq-A-form}:
\begin{equation}
\Delta_{\mathrm{LSQ}}
=
\Delta
\sqrt{
\frac{\mathbf{A}:\mathbf{A}}
     {\mathbf{A}:\mathbf{A}}
}
=
\Delta~,
\qquad
\mathbf{A}:\mathbf{A}>0~.
\label{eq:isotropic_lsq}
\end{equation}
Hence, metric compression alone removes the dependence of the LSQ filter--width on the local resolved kinematics. The resulting scale is robust with
respect to perturbations of the directional mesh metrics, but is purely
geometric.
This observation identifies two distinct requirements for the compressed formulation. First, the LSQ normalization must remain well defined as $\mathbf{A}\rightarrow\mathbf{0}$. Second, a non-trivial dependence on the resolved velocity gradient structure must be retained to preserve the flow-dependent character of the original LSQ construction. \\
In this context, a generic denominator regularization does not necessarily satisfy both requirements.
Indeed, a first possibility is a scalar safeguard of the contracted denominator: $\max\left(\mathbf{A}:\mathbf{A},\varepsilon_G^2\right)$~, where $\varepsilon_G>0$ is a prescribed regularization threshold having the same
physical dimensions as the entries of $\mathbf{A}$. 
The resulting filter width is
\begin{equation}
    \Delta_{\mathrm{LSQ}}
    =
    \Delta
    \sqrt{
        \frac{\mathbf{A}:\mathbf{A}}
        {\max\left(\mathbf{A}:\mathbf{A},\varepsilon_G^2 \right)}
    }~.
    \label{eq:scalar_regularized_filter}
\end{equation}
For $\mathbf{A}:\mathbf{A}>\varepsilon_G^2$, the regularized expression reduces to $\Delta_{\mathrm{LSQ}}=\Delta$. It follows that, away from regions where the velocity gradient magnitude is comparable to the regularization threshold, the scalar safeguard reduces to the purely metric filter. The local kinematic modulation is therefore lost.
A second possibility is a componentwise magnitude based regularization:
\begin{equation}
    \left[
        \mathcal{S}^{\mathrm{abs}}_{\varepsilon_G}(\mathbf{A})
    \right]_{ij}
    =
    \max\left(\left|\mathbf{A}\right|_{ij},\varepsilon_G\right),
    \label{eq:absolute_regularization}
\end{equation}
leading to
\begin{equation}
    \Delta_{\mathrm{LSQ}}
    =
    \Delta
    \sqrt{
        \frac{\mathbf{A}:\mathbf{A}}
        {
        \mathcal{S}^{\mathrm{abs}}_{\varepsilon_G}(\mathbf{A})
        :
        \mathcal{S}^{\mathrm{abs}}_{\varepsilon_G}(\mathbf{A})
        }
    }~.
    \label{eq:absolute_regularized_filter}
\end{equation}
%
Therefore, when all the entries of $\mathbf{A}$ are sufficiently larger in magnitude
than $\varepsilon_G$, this regularization recovers exactly
\begin{equation}
\mathcal{S}^{\mathrm{abs}}_{\varepsilon_G}(\mathbf{A})
:
\mathcal{S}^{\mathrm{abs}}_{\varepsilon_G}(\mathbf{A})
=
\mathbf{A}:\mathbf{A}~,
\end{equation}
and $\Delta_{\mathrm{LSQ}}=\Delta$.
Thus, the magnitude based regularization removes the singularity but, similarly to
the scalar safeguard, provides only a weak kinematic modulation over most of the domain.\\
%
The above constructions regularize the LSQ normalization but do not satisfy the second requirement for the filter--width. For this reason, we seek a regularized normalization that remains strictly positive while preserving a non-trivial dependence on the local component structure of $\mathbf{A}$:
%
%
\begin{equation}
    \left[
        \mathcal{S}_{\varepsilon_G}(\mathbf{A})
    \right]_{ij}
    =
    \max\left(A_{ij},\varepsilon_G\right)~.
    \label{eq:sign_selective_regularization}
\end{equation}
%
%
Although $\mathbf{A}$ is symmetric positive semidefinite, its off--diagonal entries may
be negative.  
In particular, $A_{ij}= \boldsymbol{g}_i\cdot\boldsymbol{g}_j$~,
where $\boldsymbol{g}_i$ and $\boldsymbol{g}_j$ denote directional velocity gradient
vectors. Therefore, negative off-diagonal entries correspond to local directional anti--correlations. The mapping
$\mathcal{S}_{\varepsilon_G}$ replaces these negative cross-correlations by the
positive floor $\varepsilon_G$ in the denominator, while their complete squared
contribution is retained in the numerator through $\mathbf{A}:\mathbf{A}$~.
It is important to remark that the operation in eq.~\ref{eq:sign_selective_regularization} is performed
componentwise in the Cartesian frame used to define the directional mesh scales
and to assemble the discrete velocity gradient tensor.\\
Thus, the resulting regularized compressed LSQ filter is defined as
\begin{equation}
    \Delta_{\mathrm{LSQ,c}}
    =
    \sqrt{
        \frac{
            (\Delta \mathbf{A}):(\Delta \mathbf{A})
        }{
            \mathcal{S}_{\varepsilon_G}(\mathbf{A})
            :
            \mathcal{S}_{\varepsilon_G}(\mathbf{A})
        }
    }~.
    \label{eq:regularized_compressed_lsq}
\end{equation}
Since $\Delta$ is a scalar quantity, eq.~\ref{eq:regularized_compressed_lsq}
admits the exact factorization
\begin{equation}
    \Delta_{\mathrm{LSQ,c}}
    =
    \Delta~\mathrm{R}(  \mathbf{G})~,
    \label{eq:factorized_filter}
\end{equation}
where
\begin{equation}
    \mathrm{R}(\mathbf{G})
    =
    \sqrt{
        \frac{
            \mathbf{A}:\mathbf{A}
        }{
            \mathcal{S}_{\varepsilon_G}(\mathbf{A})
            :
            \mathcal{S}_{\varepsilon_G}(\mathbf{A})
        }
    }~.
    \label{eq:kinematic_factor}
\end{equation}
The two factors in eq.~\ref{eq:factorized_filter} perform distinct and
complementary functions. The max--min metric scale $\Delta$ controls the
transmission of perturbations affecting the directional mesh metrics,
whereas $\mathrm{R}(\mathbf{G})$ provides a residual dependence on the local
resolved kinematics. The regularized denominator simultaneously prevents
singular behaviour.
It is also very important to note that the kinematic factor is globally bounded. Indeed, recalling the positive-semidefinite character of $\mathbf{A}$, its entries satisfy $A_{ij}^2 \leq A_{ii}A_{jj}$.
It follows that
\begin{equation}
    \mathbf{A}:\mathbf{A}
    \leq
    \left(\operatorname{tr}\mathbf{A}\right)^2
    \leq
    3\sum_{i=1}^{3}A_{ii}^2~;
    \label{eq:numerator_bound}
\end{equation}
moreover, since $A_{ii}\geq 0$,
\begin{equation}
    \mathcal{S}_{\varepsilon_G}(\mathbf{A})
    :
    \mathcal{S}_{\varepsilon_G}(\mathbf{A})
    \geq
    \sum_{i=1}^{3}A_{ii}^2~.
    \label{eq:denominator_bound}
\end{equation}
Combining eqs.~\ref{eq:numerator_bound} and
\ref{eq:denominator_bound} gives
\begin{equation}
    0\leq \mathrm{R}(\mathbf{G})\leq\sqrt{3}~\Longrightarrow
    0
    \leq
    \Delta_{\mathrm{LSQ,c}}
    \leq
    \sqrt{3}\,\Delta~.
\end{equation}
Thus, the sign-selective regularization satisfies both requirements identified in this subsection.
It is also important to note that the factorized form
in eq.~\ref{eq:factorized_filter} provides a direct
estimate of the sensitivity of the complete formulation to perturbations of the
directional mesh metrics. For a fixed resolved velocity gradient tensor, the
factor $\mathrm{R}(\mathbf{G})$ is unchanged by such perturbations. Hence,
\begin{equation}
    \frac{
        \widehat{\Delta}_{\mathrm{LSQ,c}}
    }{
        \Delta_{\mathrm{LSQ,c}}
    }
    =
    \frac{
        \widehat{\Delta}
    }{
        \Delta
    }~,
    \qquad \mathrm{R}(\mathbf{G})>0~,
    \label{eq:complete_filter_perturbation}
\end{equation}
and therefore
\begin{equation}
    \log
    \left(
        \frac{
            \widehat{\Delta}_{\mathrm{LSQ,c}}
        }{
            \Delta_{\mathrm{LSQ,c}}
        }
    \right)
    =
    \log
    \left(
        \frac{
            \widehat{\Delta}
        }{
            \Delta
        }
    \right)~.
    \label{eq:complete_filter_log_perturbation}
\end{equation}
Thus, for frozen resolved kinematics, the transmission of metric perturbations is
governed entirely by the max--min metric mapping, whereas $\mathrm{R}(\mathbf{G})$ provides the
bounded kinematic modulation required to retain the physical effectiveness of the
flow-dependent filter. The complete formulation defined by
eqs.~\ref{eq:factorized_filter}--\ref{eq:kinematic_factor} is hereafter referred
to as the max--min compressed LSQ filter, or LSQ--maxmin for brevity.\\
Finally, in all the computations presented below, the regularization threshold is defined as $\varepsilon_G  =\varepsilon_G^{*} \left({U_{\mathrm{ref}}}/{L_{\mathrm{ref}}}\right)^2$~, where $\varepsilon_G^{*}=10^{-8}$ is kept fixed for all flow configurations and computing architectures. $U_{\mathrm{ref}}$ and $L_{\mathrm{ref}}$ denote the characteristic velocity and length scales of the considered configuration.
%
\subsection{Componentwise standard LSQ filter regularization}
\noindent A relevant implementation detail concerns the denominator regularization adopted in the standard LSQ computations. Specifically, all the standard LSQ implementations considered in this work employ the
componentwise regularization introduced in eq.~\ref{eq:sign_selective_regularization}~:
\begin{equation}
    \Delta_{\mathrm{LSQ},\varepsilon}
    =
    \sqrt{
        \frac{
            \left(\boldsymbol{\Delta}_{D}\mathbf{A}\right)
            :
            \left(\boldsymbol{\Delta}_{D}\mathbf{A}\right)
        }{
            \mathcal{S}_{\varepsilon_G}(\mathbf{A})
            :
            \mathcal{S}_{\varepsilon_G}(\mathbf{A})
        }
    }~.
    \label{eq:lsq_regularized_baseline}
\end{equation}
This regularization procedure modifies the absolute value of the standard filter  width, but it does not alter the transmission of directional metric
perturbations analyzed above when the resolved kinematics are held fixed.
Indeed, if $\mathbf{A} : \mathbf{A}>0$, eq.~\ref{eq:lsq_regularized_baseline} can be factorized exactly
as
\begin{equation}
    \Delta_{\mathrm{LSQ},\varepsilon}
    =
    \Delta_{\mathrm{LSQ}}~
    \mathrm{R}_{\varepsilon_G}(\mathbf{G})~,
    \label{eq:lsq_regularized_factorization}
\end{equation}
where
\begin{equation}
    \mathrm{R}_{\varepsilon_G}(\mathbf{G})
    =
    \sqrt{
        \frac{
            \mathbf{A}:\mathbf{A}
        }{
            \mathcal{S}_{\varepsilon_G}(\mathbf{A})
            :
            \mathcal{S}_{\varepsilon_G}(\mathbf{A})
        }
    }~.
    \label{eq:R_regularized_baseline}
\end{equation}
The factor \(\mathrm{R}_{\varepsilon_G}(\mathbf{G})\) depends exclusively on the
resolved velocity gradient tensor and is therefore unchanged by
perturbations affecting only the directional mesh scales. Hence,
for fixed \(\mathbf{G}\),
\begin{equation}
    \log\left(
        \frac{
            \widehat{\Delta}_{\mathrm{LSQ},\varepsilon}
        }{
            \Delta_{\mathrm{LSQ},\varepsilon}
        }
    \right)
    =
    \log\left(
        \frac{
            \widehat{\Delta}_{\mathrm{LSQ}}
        }{
            \Delta_{\mathrm{LSQ}}
        }
    \right).
    \label{eq:regularized_same_log_transmission}
\end{equation}
Consequently, all the metric perturbation bounds derived above for the
standard LSQ formulation carry over unchanged to the
componentwise regularized implementation. 
%
\section{Filter--width implementation strategies}
\label{sec:impl}

\tikzset{
	implbox/.style={draw, rounded corners=2pt, align=center, minimum height=0.85cm,
		inner sep=5pt},
	implflow/.style={-stealth, line width=0.7pt},
	implnote/.style={font=\scriptsize, align=center}
}
The filter--width formulations discussed in Sec.~\ref{sec:gov} were implemented
in SPUMA~\cite{bna2025spuma}, which preserves the OpenFOAM finite--volume
data layout and exposes parallel execution through a portable programming model.
Consistently with its minimally invasive design, computational loops in SPUMA
retain the face--based structure of the corresponding
OpenFOAM implementations. This aspect is particularly relevant for the evaluation
of the directional length scales in eq.~\ref{eq:deltaxj}, whose cell--centred
values are obtained by accumulating the projected area contributions of the
surrounding faces.\\
In the following, the standard face--based GPU implementation is denoted as
\textit{LSQ--GPU (face--based)}. Moreover, to separate the effect of the filter
formulation from that of the accumulation strategy, a deterministic cell--based
variant of the same standard LSQ formulation is also considered and denoted as
\textit{LSQ--GPU (cell--based)}. Finally, the face--based implementation of the
LSQ--maxmin formulation is denoted as \textit{LSQ--maxmin--GPU}. When the same
bounded formulation is executed on the CPU reference backend, it is denoted as
\textit{LSQ--maxmin--CPU}.
\subsection{Face--based accumulation}
In the standard SPUMA implementation, the directional metric accumulation retains the face--based structure of the corresponding OpenFOAM operator.
The loop is parallelized over mesh faces. For each face, the absolute Cartesian
components of the face area vector are added to the directional sums of the owner
cell and, for internal faces, to those of the neighbour cell. On GPUs, different
threads can update the same cell simultaneously resulting in a race condition; the conflict is removed by
using \texttt{foamAtomic::AtomicAdd}. The data flow and the two possible cell
updates are expressed in the pseudocode of Algorithm~\ref{alg:face-based}.\\
%
%
\begin{algorithm}[htbp]
\caption{Face-based accumulation used by LSQ--GPU (face--based).}
\label{alg:face-based}
\begin{algorithmic}[1]
\Require Face vector areas $\mathbf{S}_f$, owner and neighbour addressing,
	number of internal faces $n_{\mathrm{intf}}$
\Ensure Cell-centred projected-area sums $\mathcal{A}_{c,i}$
\ForAll{$f \in \{0,\ldots,n_{\mathrm{faces}}-1\}$}
	\State $(a_x,a_y,a_z) \gets (|S_{f,x}|,|S_{f,y}|,|S_{f,z}|)$
	\State $P \gets \mathrm{owner}[f]$
	\For{$i \in \{x,y,z\}$}
		\State \Call{foamAtomic::AtomicAdd}{$\mathcal{A}_{P,i},a_i$}
	\EndFor
	\If{$f < n_{\mathrm{intf}}$}
		\State $N \gets \mathrm{neighbour}[f]$
		\For{$i \in \{x,y,z\}$}
			\State \Call{foamAtomic::AtomicAdd}{$\mathcal{A}_{N,i},a_i$}
		\EndFor
	\EndIf
\EndFor
\end{algorithmic}
\end{algorithm}
This implementation is compact and directly mirrors the original finite--volume
face loop. Its drawback is that the effective ordering of the atomic updates is not prescribed. Since floating--point addition is not associative, different accumulation orders may introduce small round--off variations in the projected face--area sums entering eq.~\ref{eq:deltaxj}. These perturbations remain at the level of the metric evaluation, but may subsequently be transmitted through the filter--width formulation.
The \textit{LSQ--maxmin--GPU} implementation keeps this same face--based accumulation
but computes a bounded scalar filter width after the projected face--area sums
have been accumulated. The operations are given in Algorithm~\ref{alg:maxmin}.

\begin{algorithm}[htbp]
\caption{Spatial metric compression used by LSQ--maxmin--GPU.}
\label{alg:maxmin}
\begin{algorithmic}[1]
\Require Cell volume $V_c$, cell-centred projected-area sums $\mathcal{A}_{c,x},\mathcal{A}_{c,y},
	\mathcal{A}_{c,z}$, $C_\Delta$ and number of cells $n_c$
\Ensure Bounded scalar filter width $\Delta_c$ and
	$\boldsymbol{\Delta}_{c,D}=\Delta_c\mathbf{I}$
\ForAll{$c \in \{0,\ldots,n_{\mathrm{c}}-1\}$}
\For{$i \in \{x,y,z\}$}
	\State $\Delta_{c,i} \gets 2V_c/\mathcal{A}_{c,i}$
\EndFor
\State $\Delta_g \gets (\Delta_{c,x}\Delta_{c,y}\Delta_{c,z})^{1/3}$
\State $\Delta_{\min} \gets \min(\Delta_{c,x},\Delta_{c,y},\Delta_{c,z})$
\State $\Delta_c \gets \max(\Delta_g,C_\Delta\Delta_{\min})$
\State $\boldsymbol{\Delta}_{c,D} \gets \Delta_c\mathbf{I}$
\EndFor
\end{algorithmic}
\end{algorithm}
\subsection{Cell--based CSR accumulation}
To separate the effect of the filter formulation from that of the accumulation
strategy, a cell--based implementation was introduced. This version uses the same
standard LSQ formulation as \textit{LSQ--GPU (face--based)}, but replaces the
atomic face loop by a cell--based summation over the faces surrounding each cell.
It relies on a Compressed Sparse Row (CSR) representation of the mesh
cell--to--face connectivity, implemented in SPUMA through the
\texttt{fvMeshCsrAddressing} class.
%
The CSR structure is built from the cell--to--face connectivity associated with
OpenFOAM's \texttt{lduAddressing}. Let $n_c$ denote the number of faces incident
to cell $c$. A prefix sum of these counts defines the CSR offsets, with
$o_0=0$ and $o_{c+1}=o_c+n_c$; consequently, the total number of stored
cell--face entries is $o_{n_{\mathrm{cells}}}$. The flat face-index array
$I_k$ and orientation-sign array $\sigma_k$ are then filled by traversing the
cell--to--face connectivity. For each face entry of cell $c$, the corresponding
face index is stored in the row $k\in[o_c,o_{c+1})$, and the orientation sign
is assigned according to
\begin{equation}
\sigma_k =
\begin{cases}
+1, & I_k\in\mathcal{F}_{\mathrm{int}} \text{ and } \operatorname{owner}(I_k)=c,\\
-1, & I_k\in\mathcal{F}_{\mathrm{int}} \text{ and } \operatorname{owner}(I_k)\neq c,\\
+1, & I_k\in\mathcal{F}_{\mathrm{bnd}}.
\end{cases}
\end{equation}
Thus, the faces associated with any cell are available as one contiguous CSR
row, together with their orientation relative to that cell.

\begin{figure}[htbp]
\centering
\includegraphics[width=0.95\textwidth]{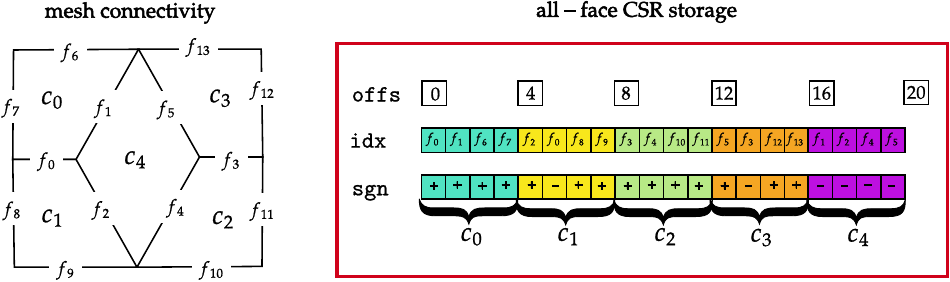}
\caption{Example of the cell--to--face CSR layout. The flat face-index and sign
arrays are partitioned according to the cell offsets.}
\label{fig:mesh5csr}
\end{figure}

The all-face CSR representation is used because boundary faces contribute to
eq.~\ref{eq:deltaxj}. The orientation signs are retained for operators that
require them, although the present accumulation involves absolute face-area
components and therefore does not depend on face orientation.
%
With this addressing, the GPU kernel is launched over cells rather than faces.
Each thread owns one cell, traverses its local CSR row in a fixed order, and
accumulates the three projected face--area sums in private registers before
writing the final values to memory. No more than one thread writes to the same cell entry,
so no atomic operation is required. Consequently, the non--deterministic
accumulation order associated with concurrent atomic updates in the face--based
implementation is eliminated. The corresponding pseudocode is given in
Algorithm~\ref{alg:cell-based}.
%
\begin{algorithm}[htbp]
\caption{Cell--based deterministic accumulation used by LSQ--GPU (cell--based).}
\label{alg:cell-based}
\begin{algorithmic}[1]
\Require CSR offsets $o_c$, face indices $I_k$, face vector areas $\mathbf{S}_f$ and number of cells $n_{\mathrm{c}}$
\Ensure Cell-centred projected-area sums $\mathcal{A}_{c,i}$
\ForAll{$c \in \{0,\ldots,n_{\mathrm{c}}-1\}$}
\State $(\mathcal{A}_{c,x},\mathcal{A}_{c,y},\mathcal{A}_{c,z})
\gets (0,0,0)$
\For{$k \in {o_c,\ldots,o_{c+1}-1}$}
\State $f \gets I_k$
\For{$i \in {x,y,z}$}
\State $\mathcal{A}_{c,i} \gets \mathcal{A}_{c,i}+|S_{f,i}|$
\EndFor
\EndFor
\EndFor
\end{algorithmic}
\end{algorithm}
The comparison between the two strategies is summarized in
Table~\ref{tab:implementation}. The face--based algorithm preserves the native
OpenFOAM ownership addressing and can therefore be directly applied to many
existing operators in SPUMA. On GPUs, however, shared cell updates require atomic
accumulations, whose effective ordering is not prescribed. In contrast, the
cell--based CSR algorithm removes this source of perturbation from the evaluation
of the directional length scales, at the cost of an additional mesh-addressing
structure and a less direct mapping to the original face--loop formulation.
Hence, in the present work, the cell--based implementation of the standard LSQ
filter is used as a numerical control to isolate the effect of the accumulation
strategy while leaving the filter formulation unchanged.
%
\begin{table}[htbp]
\centering
\caption{Filter--width implementation strategies compared in
\texttt{caaspuma}.}
\label{tab:implementation}
\small
\renewcommand{\arraystretch}{1.15}
\begin{tabularx}{0.98\textwidth}{
    @{}
    >{\raggedright\arraybackslash}p{0.29\textwidth}
    >{\raggedright\arraybackslash}p{0.18\textwidth}
    >{\raggedright\arraybackslash}X
    @{}
}
\hline
Implementation & Metric assembly & Role in this work \\
\hline
LSQ--GPU (face--based)
& Face--based with \texttt{AtomicAdd}
& Standard GPU baseline \\

LSQ--GPU (cell--based)
& Cell--based, CSR
& Diagnostic control for accumulation--order effects \\

LSQ--maxmin--GPU
& Face--based with \texttt{AtomicAdd}
& Perturbation--aware formulation \\

LSQ--maxmin--CPU
& Face--based, serial
& Cross--architecture assessment of perturbation--aware formulation \\
\hline
\end{tabularx}
\end{table}

\section{Results}\label{sec:results} 
The numerical computations were carried out using two very different flow configurations: the flow past a circular cylinder at $\mathrm{Re} = 3900$ and  the Taylor--Green Vortex at $\mathrm{Re} = 1600$.\\ 
The flow past a circular cylinder was computed using \texttt{caaspuma}, a GPU enabled implementation (developed within the SPUMA framework) of our existing low--dissipative compressible density--based solver,~\cite{dalessandro2020caafoam}.
The Taylor--Green vortex was simulated using a SPUMA version of
\texttt{rhoPimpleFoam} adopting the numerical settings of the
well established benchmark configuration reported by Zirwes et al.,~\cite{Zirwes2023}. This choice was made to provide a reproducible and
well characterized reference for the present computations.\\
%
The GPU computations were performed on four EuroHPC supercomputers:
LEONARDO at CINECA (Italy), using NVIDIA A100; JUPITER at the
JSC (Germany), using NVIDIA GH200; MareNostrum~5 at BSC (Spain), 
using NVIDIA H100; and LUMI at CSC
(Finland), using AMD Instinct MI250X.
On the NVIDIA based platforms, the code was compiled using the
\texttt{nvc++} compiler provided with the CUDA Toolkit and OpenMPI.
On LUMI, the HIP/ROCm programming environment was used together
with the AMD Clang compiler and Cray MPICH.
CPU computations were carried out using Intel Xeon
Platinum 8952+ processors on the CRESCO8 system at the ENEA site
in Portici (Italy). The codes were built using the 
Intel software stack, including the Intel
compiler suite and IntelMPI.
\subsection{Flow past a circular cylinder at $\mathrm{Re} = 3900$}
\noindent The present analysis focuses on the flow past a circular cylinder at a Reynolds number based on the cylinder diameter, $\mathrm{Re}=3900$, and a free--stream Mach number of $0.2$. The Prandtl number is set to $\mathrm{Pr}=0.71$, while the specific heat ratio is $\gamma=1.4$.
 This configuration represents a canonical benchmark for the assessment of LES methodologies owing to its pronounced sensitivity to both numerical and modeling errors,~\cite{kravchenko2000,parnaudeau2008}.  The flow is characterized by the laminar separation of the shear layers from the cylinder surface, followed by transition to turbulence in the near wake, see Fig.~\ref{fig:qcriterion} for a representation. 
 A broad selection of LES results available in the literature
 is summarized in Tab.~\ref{tab:literature_forces}.\\
\begin{figure}[htbp]
\centering
{\includegraphics[width=0.85\textwidth]{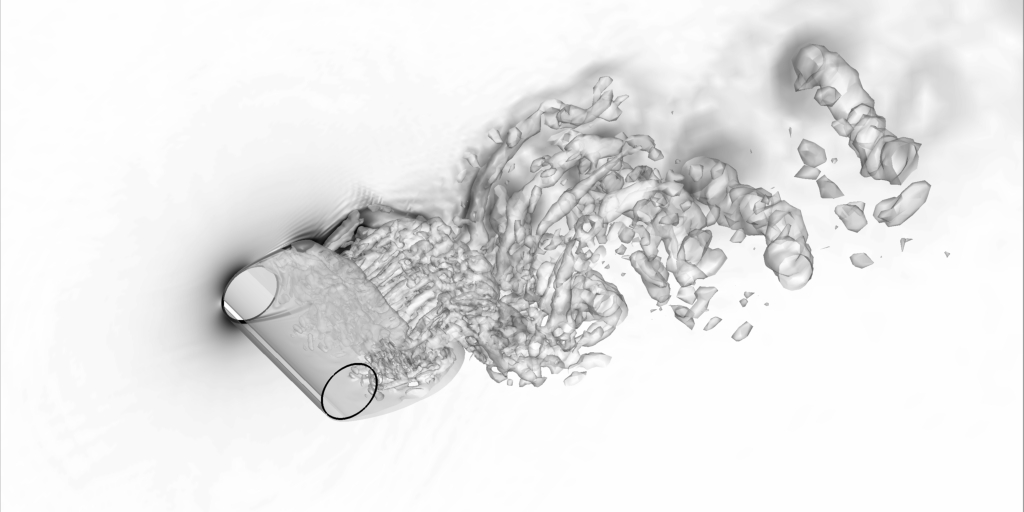} \label{fig:qcriterion-image}}
\caption{Vortical structures visualized by the Q-criterion, with the numerical Schlieren field shown on the vertical plane. The solution was obtained on the MareNostrum5 supercomputer using NVIDIA H100 accelerators. LSQ--maxmin filter width, $C_\Delta = 1.2$~.}
\label{fig:qcriterion}
\end{figure}
The computational domain extends $38D$ in the radial direction from the cylinder centre and $\pi D$ in the spanwise direction. The fully structured O--type grid was inherited directly from our previous study,~\cite{dalessandro2024}, with $N_\theta=300$, $N_r=330$, and $N_z=48$ cells in the circumferential, radial, and spanwise directions, respectively. This spanwise resolution was deliberately retained because our previous assessment identified $N_z=48$ as a discriminating configuration for the LSQ filter,~\cite{dalessandro2024}. At finer resolutions, the sensitivity to the filter--width definition decreases, and the conventional cube--root filter also recovers the correct flow behaviour. 
Computational cells were clustered near the cylinder surface, and the dimensionless height of the cell next to the wall was set to $10^{-3}$.
In addition, a sponge--layer was adopted in the far--field region as a non--reflecting boundary treatment. 
Following our previous investigations,~\cite{dalessandro2024}, the sponge layer extends $13D$ from the outer boundary, corresponding to a dimensionless thickness $L_{sp}f/c_\infty \simeq 0.5$, where $f$ is the vortex-shedding frequency and $c_\infty$ is the free--stream speed of sound. 
The target damping level was set to $\eta_{target}=40~\mathrm{dB}$ as suggested by Mani et al.~\cite{mani2012}.\\
Importantly, a well documented feature of this benchmark is the sensitivity of the near-wake  topology to numerical and modeling parameters,~\cite{beaudan1994}. The mean streamwise velocity  profile in the very near wake (at $x/D = 1.06$) can exhibit either a 
U--shaped or V--shaped character,~\cite{parnaudeau2008,Lourenco:1993}. The U--shaped profile has been confirmed as the correct physical behavior, while the V--shaped profile has  been associated with premature laminar-to-turbulent transition in the separated shear layers, either due to upstream disturbances in experiments or excessive numerical dissipation in simulations,~\cite{parnaudeau2008}.\\
The numerical simulations were performed using \texttt{caaspuma}, as introduced above. Convective terms were discretized using Pirozzoli's second-order accurate, energy conserving scheme,~\cite{pirozzoli2011}, while standard central differences were employed for the diffusive contributions. Time integration was carried out using a compact storage $(2N)$ explicit Runge--Kutta scheme. Specifically, a five-stage, fourth-order accurate method was adopted,~\cite{kennedy2000}. This combination provides a low--dissipation numerical framework suitable for scale resolving simulations, while limiting the contamination of the resolved turbulent scales by spurious numerical dissipation.
All the computations were performed using two GPU accelerators on the GPU based platforms and 384 CPU cores on CRESCO8, thereby retaining a highly parallel execution environment also for the CPU reference.\\
The numerical simulations were first advanced for $t^* \simeq 220$,
where $t^*=tU_{\mathrm{ref}}/L_{\mathrm{ref}}$, before statistical sampling was initiated. They were then continued for an additional $\Delta t^* \simeq 235$, over which the flow statistics were accumulated. Considering the radial domain extent, the sampling interval spans approximately 6.2 flow-through times.
Moreover, the sampling interval corresponds to approximately 50 vortex--shedding cycles.\\
\begin{figure}[htbp]
 \centering
 \subfigure[Mean streamwise velocity.]
 {\includegraphics[width=0.485\textwidth]{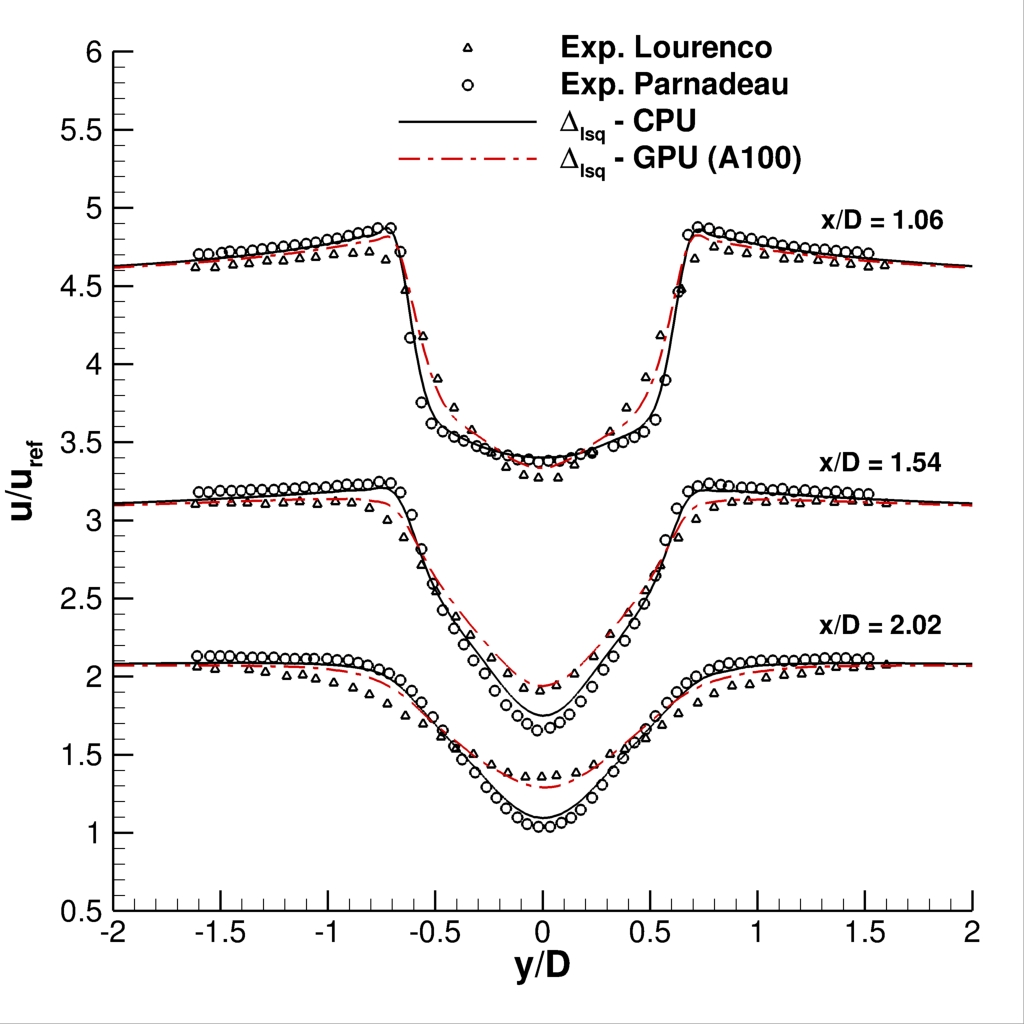} \label{fig:CPU-GPU-A100-umean}}
 \subfigure[Mean resolved streamwise Reynolds stress.]
 {\includegraphics[width=0.485\textwidth]{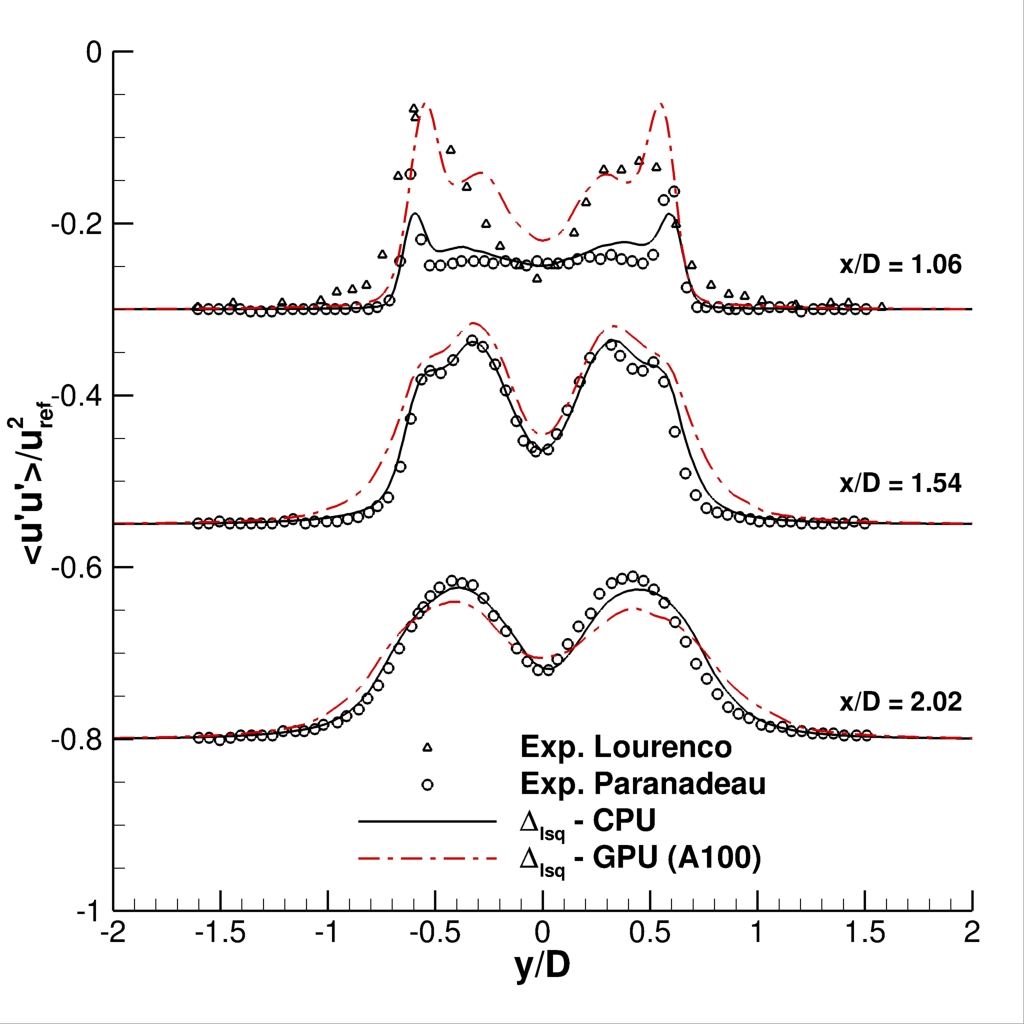}
 \label{fig:CPU-GPU-A100-uprime}}
\caption{Comparison of wake--region flow statistics obtained with the standard LSQ filter on CPU and GPU architectures. The GPU results use the face--based loop implementation.}
\label{fig:CPU-GPU-A100}
\end{figure}
We first examine the cross-architecture behaviour of the standard LSQ formulation, using the CPU solution as a reference and retaining the face-based implementation on the GPUs.
As shown in Fig.~\ref{fig:CPU-GPU-A100-umean}, the standard LSQ filter evaluated on the NVIDIA A100 produces a marked departure from the
corresponding CPU solution in the immediate near wake. In particular, at
$x/D=1.06$, the mean streamwise velocity exhibits a pronounced V--shaped
profile, which is characteristic of a premature transition of the separated
shear layers. The same behavior can be observed in Fig.~\ref{fig:vort-CPU-GPU}, where the GPU solution shows a substantially modified recirculation topology compared with the CPU result. Consistently, the resolved streamwise normal Reynolds stress reported in Fig.~\ref{fig:CPU-GPU-A100-uprime} is considerably larger for the GPU computation, especially at $x/D=1.06$.
These higher fluctuation levels indicate an earlier
breakdown of the separated shear layers.
It is important to underline that these results are consistent with the perturbations transmission mechanism identified in Sec.~\ref{sec:theory}. Indeed, in shear--dominated regions, directional metric perturbations introduced
by face--based atomic accumulations can be transmitted through the standard LSQ
operator with near-unit relative gain.
The resulting variation of the SGS activity
is consistent with the premature transition indicated by the V--shaped
velocity profile, the increased resolved Reynolds stress, and the modified
wake topology.\\
\begin{figure}[htbp]
\centering
{\includegraphics[width=0.65\textwidth]{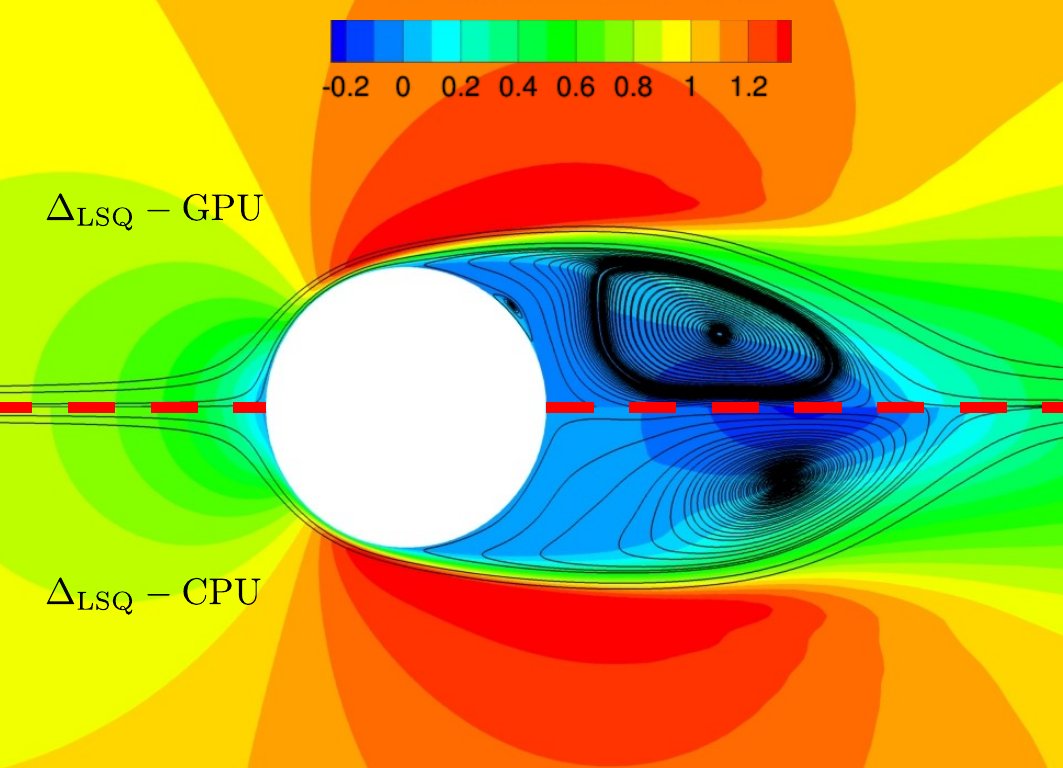} \label{fig:qcriterion-image}}
\caption{Mean velocity streamlines overlaid on the mean streamwise velocity field obtained with the standard LSQ filter. The upper and lower halves show the NVIDIA A100 GPU face--based solution and the corresponding CPU solution, respectively.}
\label{fig:vort-CPU-GPU}
\end{figure}
\begin{figure}[htbp]
 \centering
 \subfigure[Mean streamwise velocity.]
 {\includegraphics[width=0.485\textwidth]{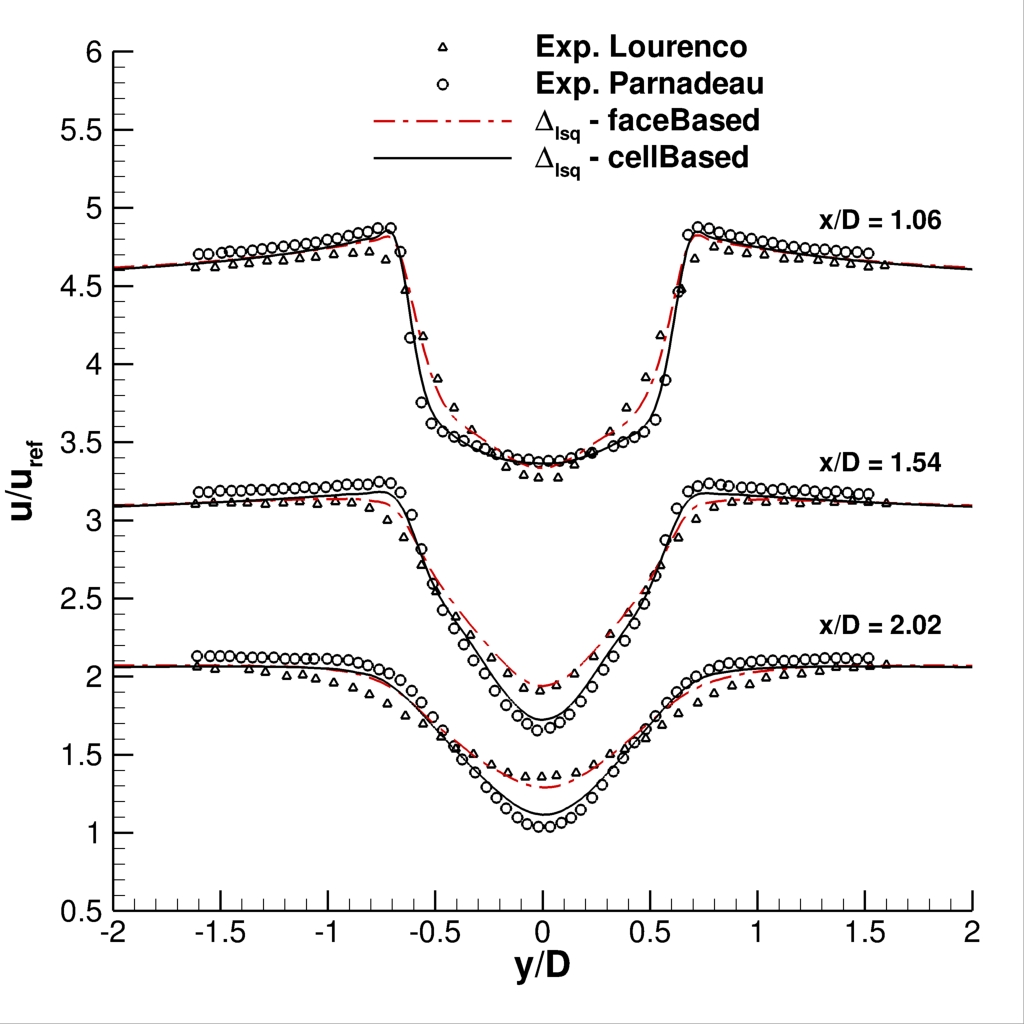} \label{fig:umean-1}}
 \subfigure[Mean resolved streamwise Reynolds stress.]
 {\includegraphics[width=0.485\textwidth]{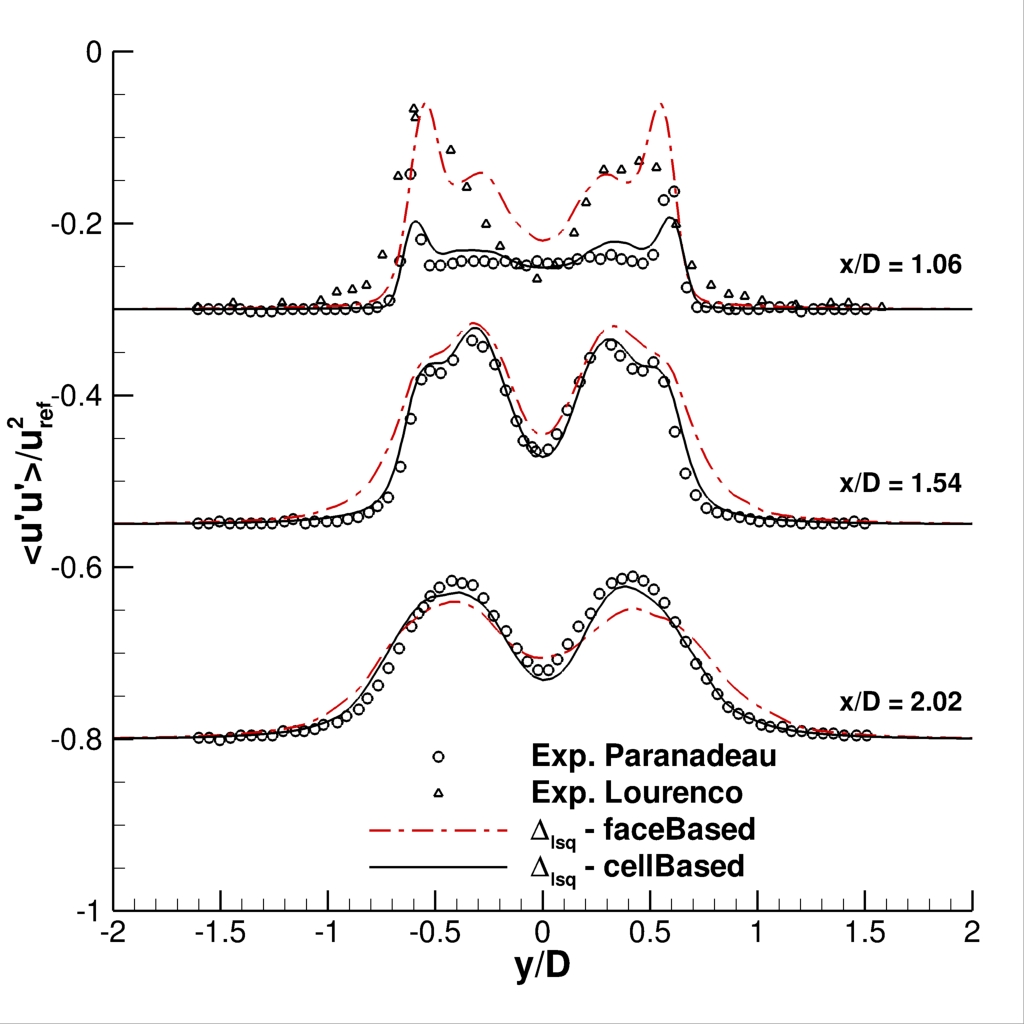}}
\caption{Comparison of wake--region flow statistics obtained with the standard LSQ filter on GPU NVIDIA A100 with different loop implementations.}
\label{fig:cell-face-A100}
\end{figure}
By contrast, Fig.~\ref{fig:cell-face-A100} isolates the effect of the metric--assembly strategy while keeping
both the GPU architecture and the standard LSQ formulation unchanged. Replacing
the face--based atomic accumulation with the deterministic cell--based loop
restores the U--shaped mean velocity profile and substantially reduces the
overprediction of the resolved streamwise Reynolds stress in the immediate
near wake. The agreement with the reference data is consequently recovered. These results further support the interpretation developed in
Sec.~\ref{sec:theory}, according to which the observed sensitivity
results from the interaction between perturbations introduced during
the assembly of the directional mesh metrics and their transmission
through the kinematic weighting of the LSQ operator.
Equivalent face--based/cell--based comparisons were carried out on NVIDIA GH200, NVIDIA H100, and AMD MI250X accelerators. All platforms reproduced the same qualitative and quantitative behaviour observed on the NVIDIA A100, with differences limited to negligible statistical variations; the corresponding results are omitted for the sake of compactness.\\
This comparison also shows that reproducibility and robustness are distinct requirements. A deterministic implementation removes perturbations caused by non-deterministic accumulation, while a robust filter-width formulation should limit their impact on the resolved flow.
%
\begin{figure}[htbp]
 \centering
 \subfigure[Mean streamwise velocity. NVIDIA A100, NVIDIA GH200, AMD MI250X; $\eta = 10^{-5}$.]
 {\includegraphics[width=0.485\textwidth]{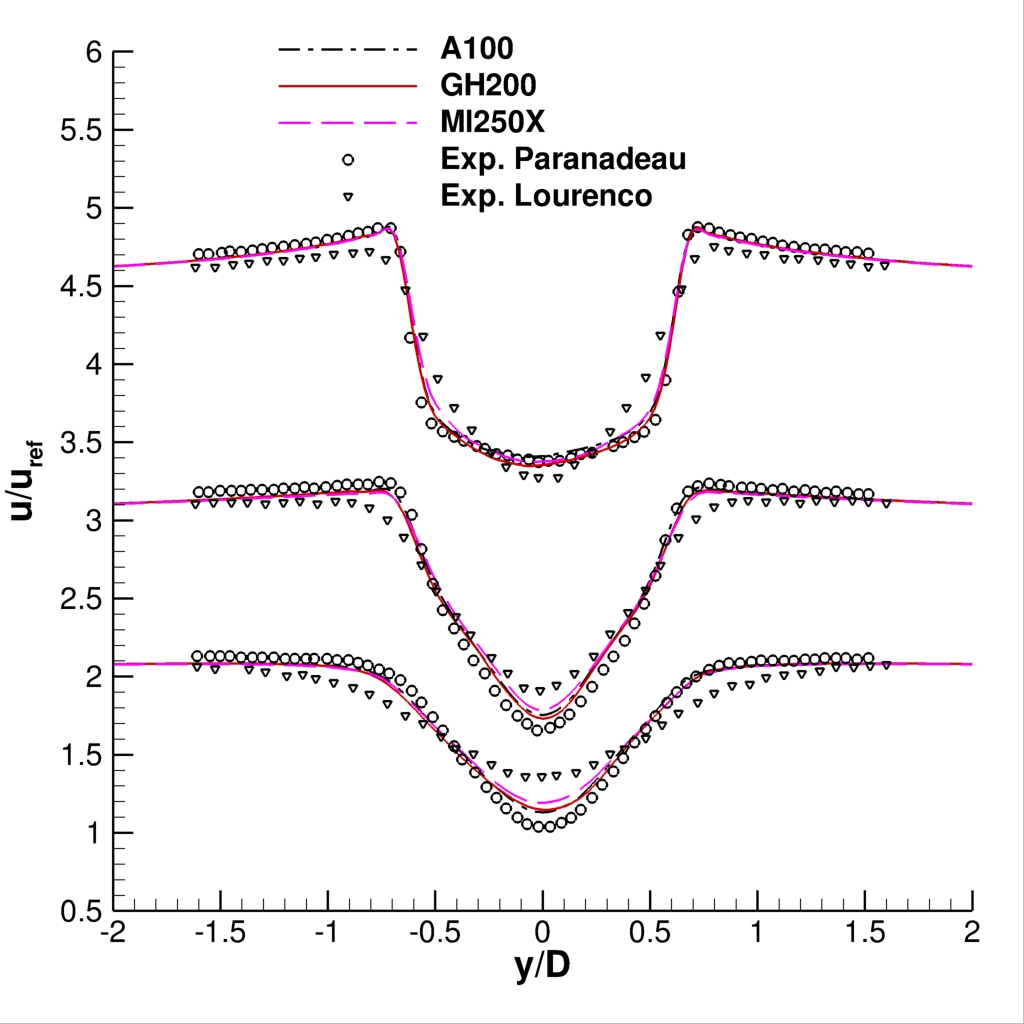} \label{fig:umean-1}}
 \subfigure[Mean streamwise velocity. NVIDIA H100.]
 {\includegraphics[width=0.485\textwidth]{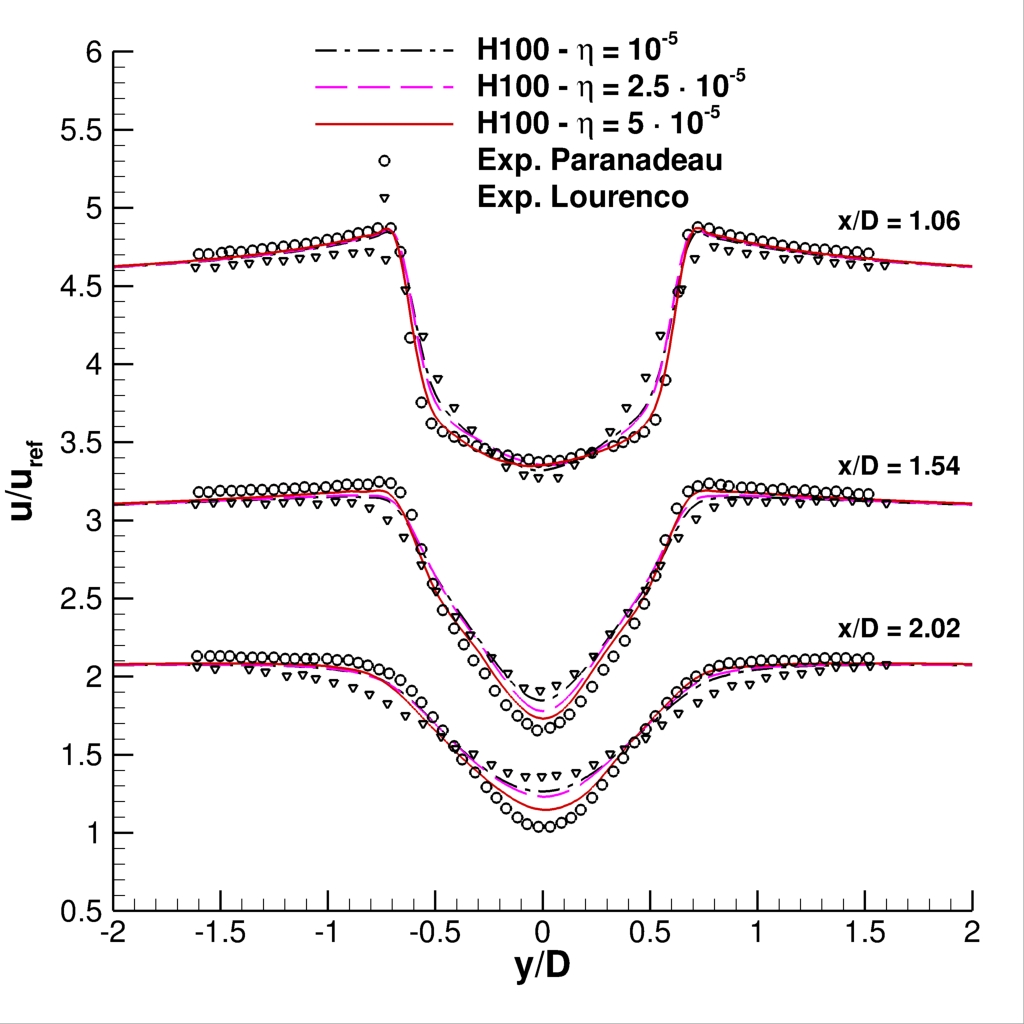}}
\caption{Effect of the Tikhonov-inspired kinematic diagonal loading
on the mean streamwise velocity profiles. Comparison among the considered architectures using standard LSQ filter with face--based implementation.}
\label{fig:kinematic-loading}
\end{figure}
Fig.~\ref{fig:kinematic-loading} shows the response of the standard face--based LSQ formulation to a controlled kinematic loading. A common diagnostic value, $\eta=10^{-5}$, was applied to all GPU architectures. This perturbation is sufficient to recover the physically consistent U-shaped near-wake profile on the NVIDIA A100, NVIDIA GH200, and AMD MI250X architectures, whereas the H100 solution remains in the V-shaped state. For the latter architecture, increasing the loading to $\eta=5\cdot10^{-5}$ recovers the U-shaped profile.
By modifying the local kinematic contribution while leaving the directional mesh metrics and their face-based assembly unchanged, the kinematic loading isolates the role of the LSQ weighting. The recovery of the U-shaped state then clearly indicates that this weighting mediates the transmission of architecture-dependent metric perturbations to the shear-layer dynamics.\\
%
%
We now examine the behaviour of the proposed LSQ--maxmin formulation across all the considered architectures in the presence of floating--point perturbations.
%
The filter--width formulation is first assessed by comparing CPU and
NVIDIA A100 computations performed using the same value
$C_\Delta=1.3$. 
%
As shown in Fig.~\ref{fig:maxmin-CPU-A100}, the CPU and GPU predictions obtained
with LSQ--maxmin are virtually indistinguishable. In particular, both
computations recover the physically consistent U-shaped mean velocity
profile at $x/D=1.06$ and provide closely matching resolved Reynolds
stresses throughout the near wake. 
In contrast to the cell--based strategy, the filter-width formulation introduced in this paper does
not eliminate the perturbation source, but controls the way in which
metric perturbations are transmitted to the SGS closure.
\begin{figure}[htbp]
 \centering
 \subfigure[Mean streamwise velocity.]
 {\includegraphics[width=0.485\textwidth]{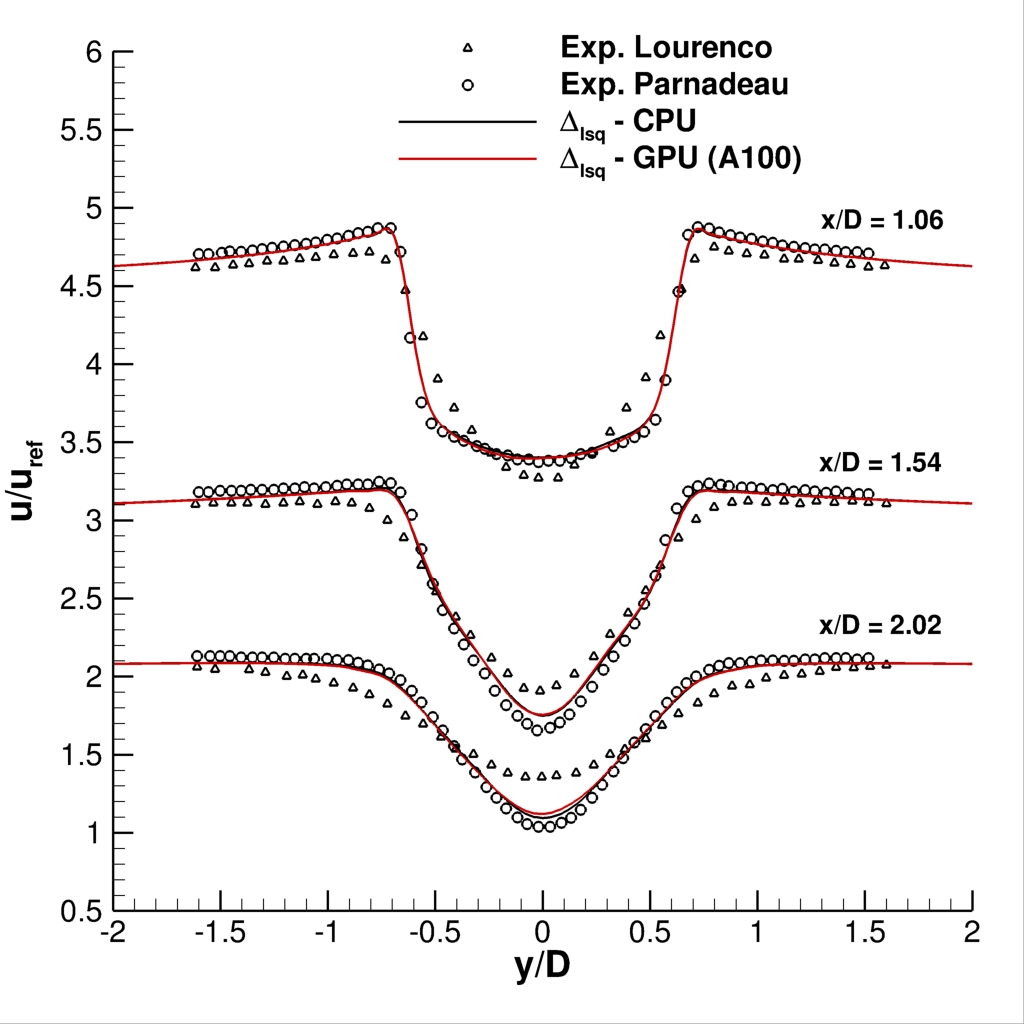} \label{fig:umean-1}}
 \subfigure[Mean resolved streamwise Reynolds stress.]
 {\includegraphics[width=0.485\textwidth]{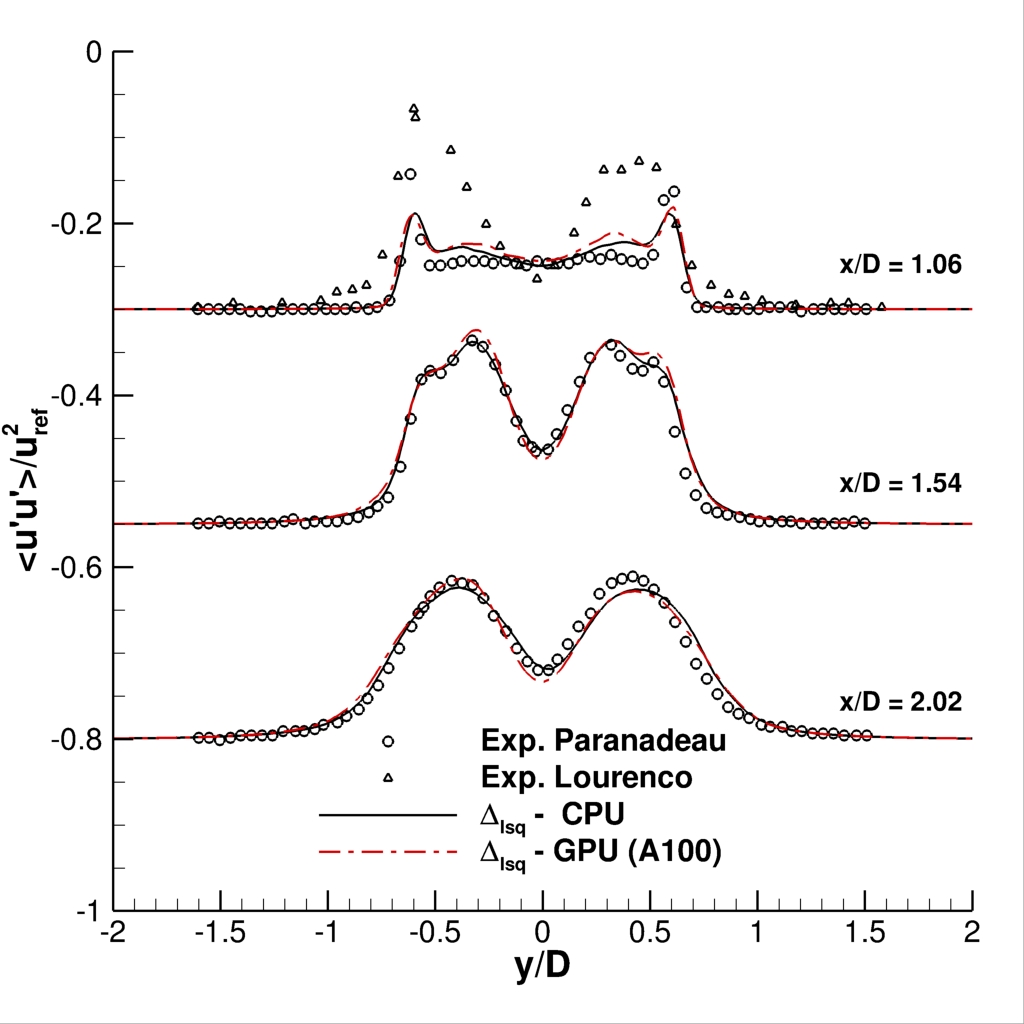}}
\caption{Comparison of wake-region flow statistics obtained with the proposed LSQ-maxmin filter on CPU and GPU architectures. The same value $C_\Delta = 1.3$ is used on both architectures.}
\label{fig:maxmin-CPU-A100}
\end{figure}
To assess the sensitivity of the proposed LSQ--maxmin formulation to the parameter $C_\Delta$, computations were performed over the range
$1.0 \leq C_\Delta \leq 1.5$ on all the architectures considered.
Only the mean streamwise velocity profile at $x/D=1.06$ is reported,
since this location provides the clearest diagnostic of the transition
state of the separated shear layers.
The limiting case $C_\Delta=1$ has a direct interpretation. Since
$\Delta_g \geq \Delta_{\min}$ for positive directional scales, the
metric contribution reduces identically to $\Delta=\Delta_g$, while
the complete filter retains the bounded kinematic modulation
$\mathrm{R}(\mathbf{G})$. 
Additional computations performed with the purely geometric compressed scale showed that removing this modulation does not recover the correct resolved-flow behaviour.
As shown in Fig.~\ref{fig:Cdelta-sweep},
the geometric mean compression is already sufficient to suppress the
architecture dependent transition toward the unphysical V-shaped
near-wake state. The resulting U-shaped profiles, however, still show
appreciable differences in depth and curvature. Thus, controlling the
transmission of metric perturbations is sufficient to recover the
correct qualitative wake state, but does not by itself ensure the
appropriate local level of SGS activity.
For $C_\Delta>1$, the minimum-scale branch becomes active wherever
$\Delta_g<C_\Delta\Delta_{\min}$, so that the metric scale is locally
set by $C_\Delta\Delta_{\min}$. Since the complete filter--width enters
quadratically in the SGS viscosity, activation of this branch directly
modifies the filter--width contribution to the local SGS activity and
can consequently alter the resolved to subgrid energy transfer in the
separated shear layers. Therefore, the lower--bound branch complements the perturbation robust geometric compression by modifying the local
filter scale in those regions where the switching condition is
satisfied.
It is important to remark that the dependence of the near-wake solution on $C_\Delta$ remains
moderate and does not exhibit a systematic monotonic trend.
The values
providing the closest agreement with the experimental profiles vary
slightly across architectures: values close to $1.3$ are preferred on
the CPU and NVIDIA A100, approximately $1.1$ on NVIDIA GH200,
approximately $1.2$ on NVIDIA H100, while values between $1.1$ and
$1.2$ are nearly indistinguishable on AMD MI250X.
These variations are
small compared with the qualitative discrepancy produced by the
standard LSQ formulation and, over part of the profile, are comparable
with the spread between the available experimental datasets.
For this reason, the present results do not support the identification of a
single universal best performing value of $C_\Delta$. They instead
identify a narrow robust interval,
$1.1 \leq C_\Delta \leq 1.3$, over which the physically consistent
U-shaped wake state is preserved on all the architectures considered.
Thus, the moderate architecture-to-architecture variation within this
interval is consistent with the role of $C_\Delta$ as a branch-selection
parameter. Furthermore, the modest variation of the best-performing $C_\Delta$ across architectures is consistent with small shifts in the local branch-switching pattern, while the robust interval remains essentially unchanged.\\
%
%
\begin{figure}[p]
    \centering
    \begin{minipage}[t]{0.46\textwidth}
        \centering
        \includegraphics[
            width=\linewidth,
            trim=2mm 2mm 2mm 2mm,
            clip
        ]{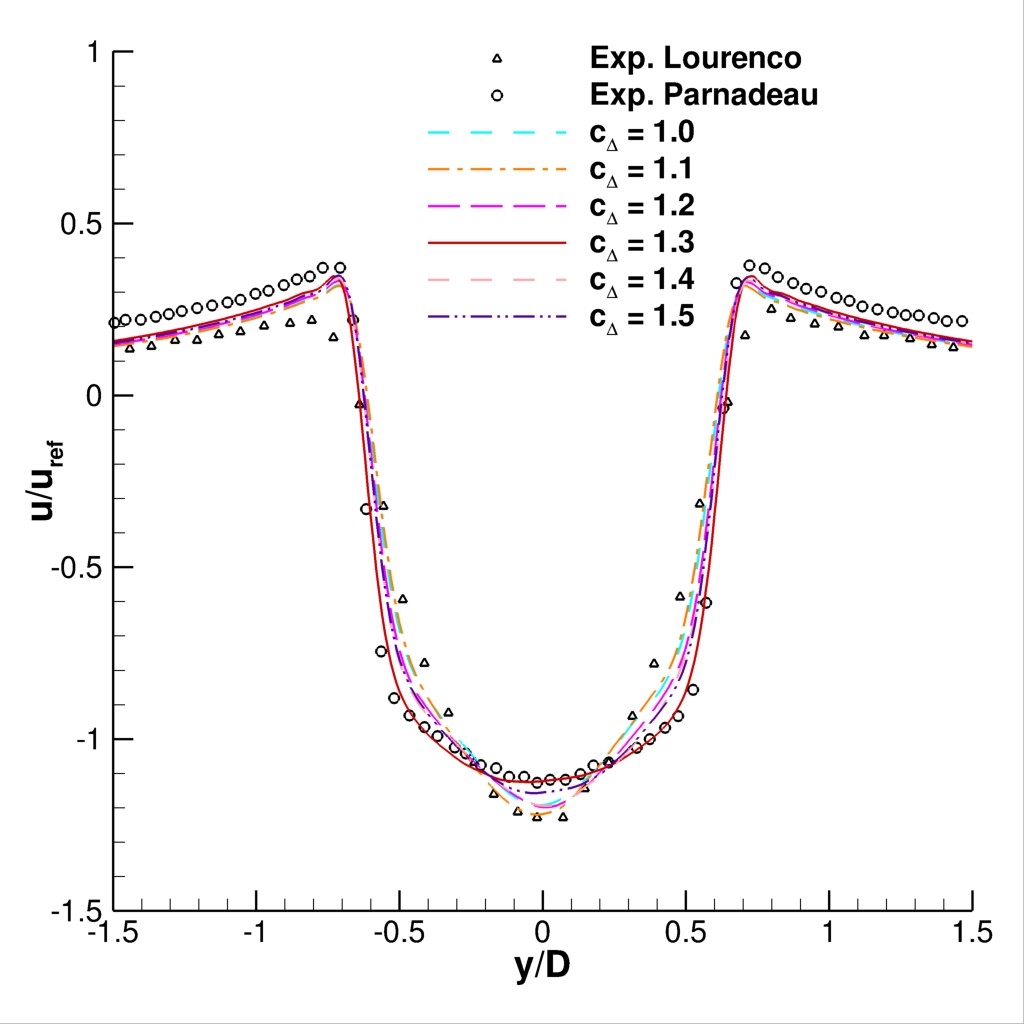}
        \vspace{-3mm}
        {\small (a) NVIDIA A100\par}
    \end{minipage}
    \hfill
    \begin{minipage}[t]{0.46\textwidth}
        \centering
        \includegraphics[
            width=\linewidth,
            trim=2mm 2mm 2mm 2mm,
            clip
        ]{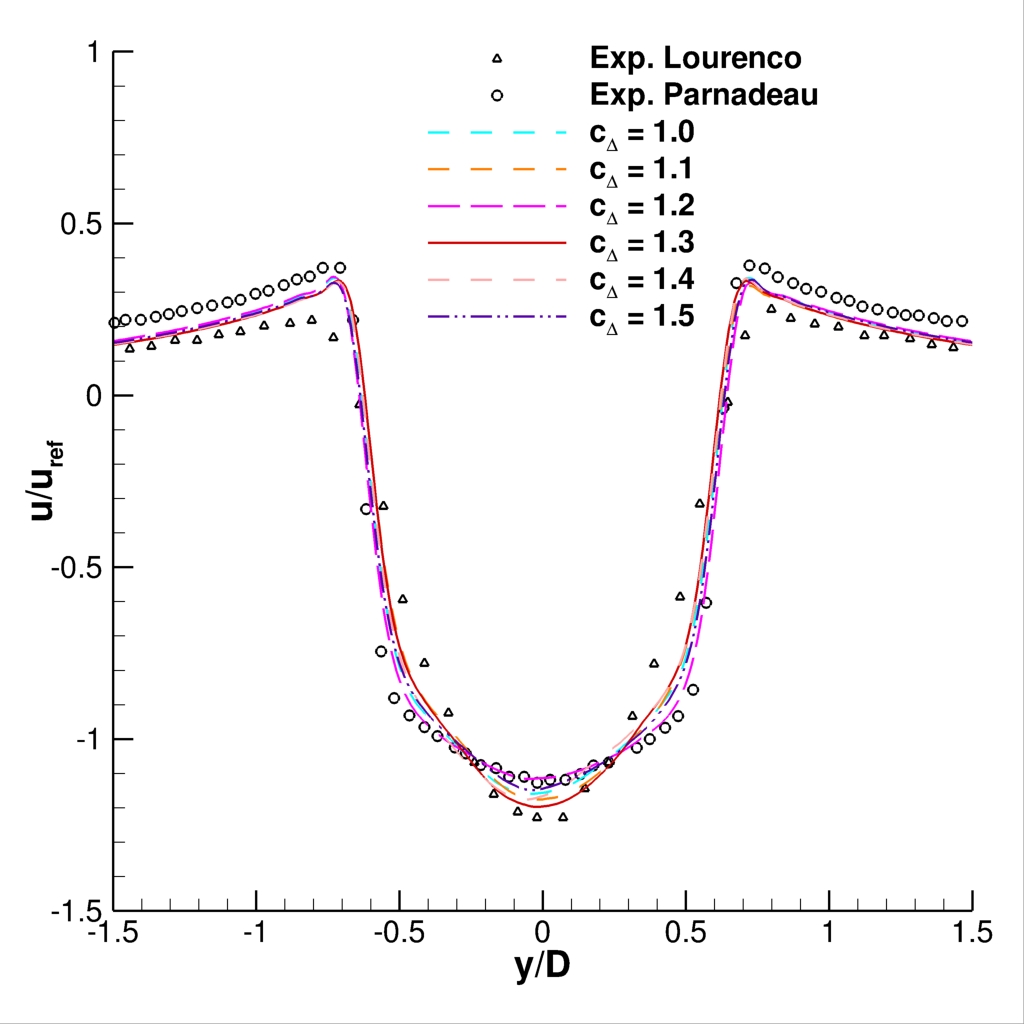}
        \vspace{-3mm}

        {\small (b) NVIDIA H100\par}
    \end{minipage}

    \vspace{-1mm}
    \begin{minipage}[t]{0.46\textwidth}
        \centering
        \includegraphics[
            width=\linewidth,
            trim=2mm 2mm 2mm 2mm,
            clip
        ]{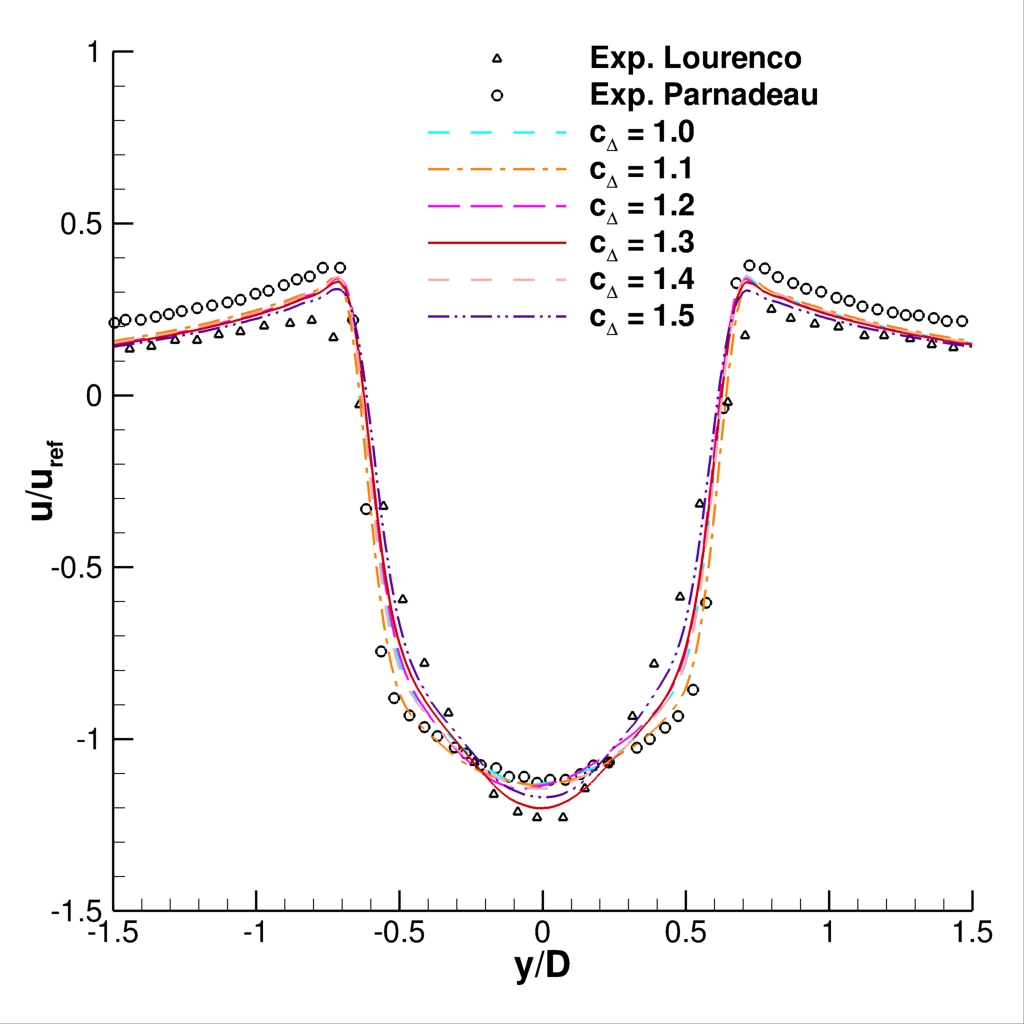}
        \vspace{-3mm}

        {\small (c) NVIDIA GH200\par}
    \end{minipage}
    \hfill
    \begin{minipage}[t]{0.46\textwidth}
        \centering
        \includegraphics[
            width=\linewidth,
            trim=2mm 2mm 2mm 2mm,
            clip
        ]{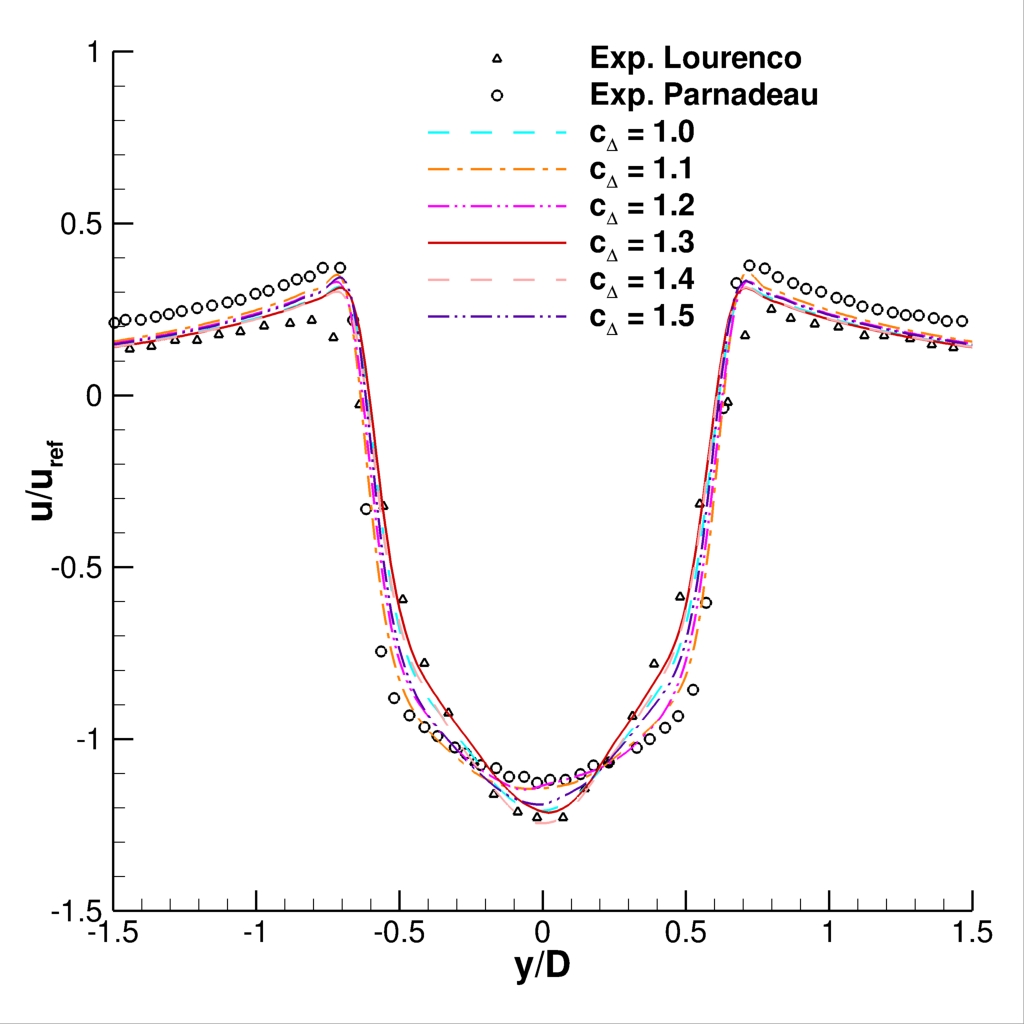}
        \vspace{-3mm}

        {\small (d) AMD MI250X\par}
    \end{minipage}

    \vspace{-1mm}

    \begin{minipage}[t]{0.46\textwidth}
        \centering
        \includegraphics[
            width=\linewidth,
            trim=2mm 2mm 2mm 2mm,
            clip
        ]{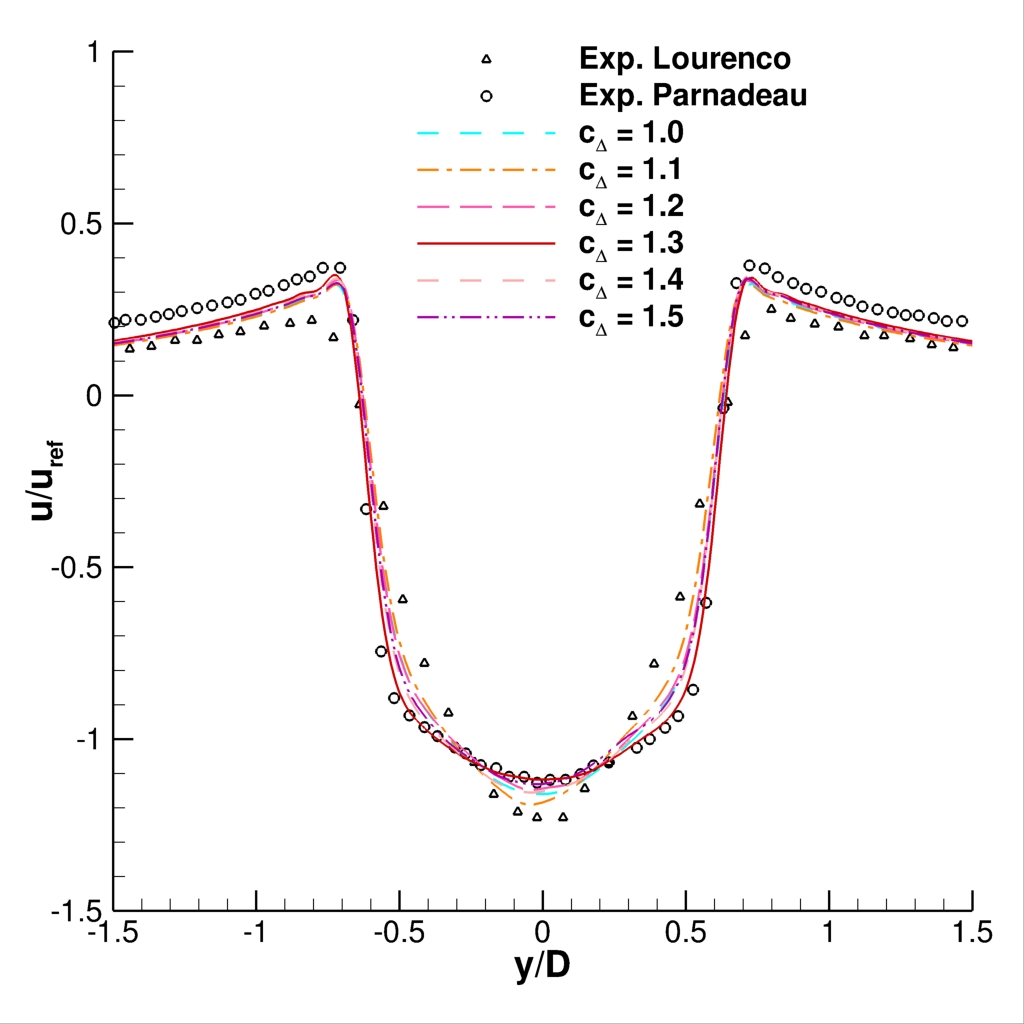}
        \vspace{-3mm}

        {\small (e) CPU\par}
    \end{minipage}

    \vspace{0mm}

    \caption{Sensitivity of the mean streamwise velocity profile at
    $x/D=1.06$ to $C_\Delta$ on the five computing architectures
    considered.}
    \label{fig:Cdelta-sweep}
\end{figure}
 \noindent Tabs.~\ref{tab:marenostrum}--\ref{tab:cresco_cpu} report the time-averaged drag coefficient, \(\langle C_D\rangle\), the root-mean-square lift coefficient, \(C_{L,\mathrm{rms}}\), and the Strouhal number, \(St\), for all the considered computing architectures. For the LSQ–maxmin formulation, the architecture-specific values of \(C_\Delta\) identified from the near-wake velocity profiles yield tightly clustered integral quantities across the five HPC systems, with \(\langle C_D\rangle \simeq 1.00\), \(C_{L,\mathrm{rms}}\simeq 0.157\), and \(St\simeq 0.214\).
%
The agreement with the available literature LES data is particularly good for $\langle C_D\rangle$ and $\mathrm{St}$, as shown in Tab.~\ref{tab:literature_forces}. A direct comparison based on $C_{L,\mathrm{rms}}$ is less conclusive, since this quantity is less systematically reported and exhibits a larger scatter among previous studies; nevertheless, the present values remain within the range of the available LES results. 
Although the integral quantities vary non-monotonically with \(C_\Delta\), this clustering provides an independent quantitative confirmation that the selected values of \(C_\Delta\) recover a common macroscopic wake state across the considered architectures.
%
Finally, it is interesting to note that our results highlight that the computations showing the incorrect V--shaped near--wake profile generally exhibit increased drag and fluctuation levels, although the trend is not uniform for all integral quantities and architectures.
%
 %
%
%
\begin{table}[htbp]
    \centering
    \caption{Flow quantities obtained on NVIDIA H100}
    \label{tab:marenostrum}

    \renewcommand{\arraystretch}{1.35}
    \setlength{\tabcolsep}{12pt}

    \begin{tabular}{|l|c|c|c|}
        \hline
        \text{Case}
        & $\langle C_D \rangle$
        & ${C_{L,\mathrm{rms}}}$
        & $\mathrm{St}$ \\
        \hline
        \noalign{\vskip 4pt}
        \hline

        Cell-based
        & 1.0187
        & 0.1882
        & 0.2185 \\
        \hline

        Face-based
        & 1.0293
        & 0.1991
        & 0.2157 \\
        \hline

        Face-based $\eta =10^{-5}$
        & 1.0392
        & 0.2236
        & 0.2130 \\
        \hline

        Face-based $\eta =2.5\cdot10^{-5}$
        & 1.0267
        & 0.1954
        & 0.2170 \\
        \hline

        Face-based $\eta = 5\cdot10^{-5}$
        & 1.0152
        & 0.1755
        & 0.2145 \\
        \hline

        $C_{\Delta}=1.0$
        & 1.0146
        & 0.1788
        & 0.2152 \\
        \hline

        $C_{\Delta}=1.1$
        & 1.0394
        & 0.2218
        & 0.2186 \\
        \hline

        $C_{\Delta}=1.2$
        & 0.9994
        & 0.1609
        & 0.2145 \\
        \hline

        $C_{\Delta}=1.3$
        & 1.0398
        & 0.2139
        & 0.2164 \\
        \hline

        $C_{\Delta}=1.4$
        & 1.0290
        & 0.2111
        & 0.2131 \\
        \hline

        $C_{\Delta}=1.5$
        & 1.0294
        & 0.2097
        & 0.2158 \\
        \hline
    \end{tabular}
\end{table}

\begin{table}[htbp]
    \centering
    \caption{Flow quantities obtained on NVIDIA GH200}
    \label{tab:jupiter}

    \renewcommand{\arraystretch}{1.35}
    \setlength{\tabcolsep}{12pt}

    \begin{tabular}{|l|c|c|c|}
        \hline
        \text{Case}
        & $\langle C_D \rangle$
        & ${C_{L,\mathrm{rms}}}$
        & $\mathrm{St}$ \\
        \hline
        \noalign{\vskip 4pt}
        \hline

        Cell-based
        & 0.9946
        & 0.1451
        & 0.2134 \\
        \hline

        Face-based
        & 1.0195
        & 0.1825
        & 0.2166 \\
        \hline

        Face-based $\eta=10^{-5}$
        & 1.0063
        & 0.1623
        & 0.2132 \\
        \hline

        $C_{\Delta}=1.0$
        & 1.0229
        & 0.1955
        & 0.2138 \\
        \hline

        $C_{\Delta}=1.1$
        & 0.9988
        & 0.1565
        & 0.2134 \\
        \hline

        $C_{\Delta}=1.2$
        & 1.0318
        & 0.2003
        & 0.2184 \\
        \hline

        $C_{\Delta}=1.3$
        & 1.0309
        & 0.2035
        & 0.2154 \\
        \hline

        $C_{\Delta}=1.4$
        & 1.0178
        & 0.1788
        & 0.2137 \\
        \hline

        $C_{\Delta}=1.5$
        & 1.0551
        & 0.2429
        & 0.2134 \\
        \hline
    \end{tabular}
\end{table}
\begin{table}[htbp]
    \centering
    \caption{Flow quantities obtained on NVIDIA A100}
    \label{tab:leonardo}

    \renewcommand{\arraystretch}{1.35}
    \setlength{\tabcolsep}{12pt}

    \begin{tabular}{|l|c|c|c|}
        \hline
        \text{Case}
        & $\langle C_D \rangle$
        & ${C_{L,\mathrm{rms}}}$
        & $\mathrm{St}$ \\
        \hline
        \noalign{\vskip 4pt}
        \hline

        Cell-based
        & 1.0012
        & 0.1587
        & 0.2137 \\
        \hline

        Face-based
        & 1.0400
        & 0.22394
        & 0.2177 \\
        \hline

        Face-based $\eta=10^{-5}$
        & 1.0285
        & 0.1971
        & 0.2182 \\
        \hline

        $C_{\Delta}=1.0$
        & 1.0569
        & 0.2464
        & 0.2156 \\
        \hline

        $C_{\Delta}=1.1$
        & 1.0639
        & 0.2579
        & 0.2185 \\
        \hline

        $C_{\Delta}=1.2$
        & 1.0383
        & 0.2141
        & 0.2126 \\
        \hline

        $C_{\Delta}=1.3$
        & 0.9999
        & 0.1538
        & 0.2137 \\
        \hline

        $C_{\Delta}=1.4$
        & 1.0247
        & 0.1933
        & 0.2144 \\
        \hline

        $C_{\Delta}=1.5$
        & 1.0218
        & 0.1891
        & 0.2137 \\
        \hline
    \end{tabular}
\end{table}
\begin{table}[htbp]
    \centering
    \caption{Flow quantities obtained on AMD MI250X}
    \label{tab:lumi}

    \renewcommand{\arraystretch}{1.35}
    \setlength{\tabcolsep}{12pt}

    \begin{tabular}{|l|c|c|c|}
        \hline
        \text{Case}
        & $\langle C_D \rangle$
        & ${C_{L,\mathrm{rms}}}$
        & $\mathrm{St}$ \\
        \hline
        \noalign{\vskip 4pt}
        \hline

        Cell-based
        & 1.0219
        & 0.1931
        & 0.2145 \\
        \hline

        Face-based
        & 1.08184
        & 0.1832
        & 0.2125 \\
        \hline

        Face-based $\eta=10^{-5}$
        & 1.0386
        & 0.2169
        & 0.2168 \\
        \hline

        $C_{\Delta}=1.0$
        & 1.0596
        & 0.2511
        & 0.2152 \\
        \hline

        $C_{\Delta}=1.1$
        & 1.0079
        & 0.1588
        & 0.2159 \\
        \hline

        $C_{\Delta}=1.2$
        & 1.0105
        & 0.1667
        & 0.2165 \\
        \hline

        $C_{\Delta}=1.3$
        & 1.0421
        & 0.2243
        & 0.217 \\
        \hline

        $C_{\Delta}=1.4$
        & 1.0724
        & 0.2142
        & 0.2142 \\
        \hline

        $C_{\Delta}=1.5$
        & 1.0302
        & 0.1986
        & 0.2145 \\
        \hline
    \end{tabular}
\end{table}
\begin{table}[htbp]
    \centering
    \caption{Flow quantities obtained on CPU}
    \label{tab:cresco_cpu}

    \renewcommand{\arraystretch}{1.35}
    \setlength{\tabcolsep}{12pt}

    \begin{tabular}{|l|c|c|c|}
        \hline
        \text{Case}
        & $\langle C_D \rangle$
        & ${C_{L,\mathrm{rms}}}$
        & $\mathrm{St}$ \\
        \hline
        \noalign{\vskip 4pt}
        \hline

        $C_{\Delta}=1.0$
        & 1.0309
        & 0.2153
        & 0.2158 \\
        \hline

        $C_{\Delta}=1.1$
        & 1.025
        & 0.1952
        & 0.2141 \\
        \hline

        $C_{\Delta}=1.2$
        & 1.0304
        & 0.2064
        & 0.2136 \\
        \hline

        $C_{\Delta}=1.3$
        & 0.998
        & 0.1544
        & 0.2137 \\
        \hline

        $C_{\Delta}=1.4$
        & 1.018
        & 0.1816
        & 0.218 \\
        \hline

        $C_{\Delta}=1.5$
        & 1.025
        & 0.1942
        & 0.2142 \\
        \hline
    \end{tabular}
\end{table}
\begin{table}[t]
\centering
    \renewcommand{\arraystretch}{1.35}
    \setlength{\tabcolsep}{12pt}
\caption{Flow past a circular cylinder at $\mathrm{Re}=3900$. LES literature data.}
\label{tab:literature_forces}
\begin{tabular}{|l|c|c|c|}
\hline 
 & $\langle C_D \rangle$ & $C_{L,\mathrm{rms}}$ & $\mathrm{St}$ \\
\hline\hline
Kravchenko and Moin,~\cite{kravchenko2000}
& 1.04  & --     & 0.210 \\
\hline
Lysenko et al., LES--SMAG~\cite{lysenko2012}
& 1.18  & 0.444  & 0.190 \\
\hline
Lysenko et al., LES--TKE~\cite{lysenko2012}
& 0.97  & 0.090  & 0.209 \\
\hline
Mittal and Moin,~\cite{Mittal:1997}
& 1.00  & --     & 0.220 \\
\hline
Meyer et al.,~\cite{meyer2010}
& 1.05  & --     & 0.210 \\
\hline
Wornom et al.,~\cite{wornom2011}
& 0.99  & 0.110  & 0.210 \\
\hline
Mani et al.,~\cite{mani2009}
& 0.99  & --     & 0.210 \\
\hline
Prsic et al.,~\cite{prsic2014}
& 1.078 & 0.1954 & 0.2154 \\
\hline
Zhang et al.,~\cite{zhang2015}
& 1.001 & 0.125  & 0.220 \\
\hline
Ouvrard et al.,~\cite{ouvrard2010}
& 0.94  & 0.092  & 0.220 \\
\hline
Franke and Frank,~\cite{franke2002}
& 0.99  & --     & 0.210 \\
\hline
Alkishriwi et al.,~\cite{alkishriwi2006}
& 1.05  & --     & 0.220 \\
\hline
Sirkar et al.,~\cite{sircar2020}
& 1.06  & --     & 0.13 \\
\hline
Tian and Xiao,~\cite{tian2020}
& 0.97  & --     & 0.1 \\
\hline
Tian and Xiao,~\cite{tian2022}
& 1.01  & --     & 0.12 \\
\hline
Shui et al.,~\cite{shui2025}
& 1.053 & 0.091 & 0.215 \\
\hline
Hao et al.,~\cite{hao2024}
& 0.995 & 0.085 & 0.203 \\
\hline
\end{tabular}
\end{table}
%

%

\subsection{Taylor--Green Vortex at $\mathrm{Re} = 1600$}
\label{sec:tgv}
\noindent The three-dimensional Taylor--Green vortex is considered as a canonical configuration for assessing LES formulations in a smooth, wall free flow condition,~\cite{Hasslberger:2021}. The initially regular array of counter-rotating vortices progressively loses coherence through vortex stretching and breakdown, generating increasingly finer scale motions and a characteristic decay of kinetic energy.
The computational domain is the periodic cube
\begin{equation}
\Omega = [-\pi L,\pi L]\times[-\pi L,\pi L]\times[-\pi L,\pi L]~,
\end{equation}
where $L$ denotes the characteristic length scale of the reference configuration. 
The corresponding velocity field is initialized as
\begin{equation}
\label{eq:tgv-initial-velocity}
\begin{aligned}
u_x
&=
u_0
\sin\left(\frac{x}{L}\right)
\cos\left(\frac{y}{L}\right)
\cos\left(\frac{z}{L}\right)~,
\\
u_y
&=
-u_0
\cos\left(\frac{x}{L}\right)
\sin\left(\frac{y}{L}\right)
\cos\left(\frac{z}{L}\right)~,
\\
u_z
&=0~,
\end{aligned}
\end{equation}
with $u_0$ denoting the characteristic velocity. This field is
divergence free by construction. By contrast, the corresponding initial pressure
distribution is
\begin{equation}
\label{eq:tgv-initial-pressure}
p 
=
p_0
+
\frac{\rho_0 u_0^2}{16}
\left[
\cos\left(\frac{2x}{L}\right)
+
\cos\left(\frac{2y}{L}\right)
\right]
\left[
\cos\left(\frac{2z}{L}\right)+2
\right]~,
\end{equation}
with $p_0$ and $\rho_0$ representing the reference pressure and density,
respectively. The Reynolds number is defined as
$\mathrm{Re}=\rho_0 u_0 L/\mu_0$,
$\mu_0$ being the reference dynamic viscosity, while the corresponding
convective time scale is $t_{\mathrm{ref}}=L/u_0$.
%
%
%
It is worth noting that in the present computations, $\mathrm{Re}=1600$ is considered, consistently with the standard incompressible TGV benchmark configuration,~\cite{Zirwes2023}. Since the governing equations are solved in compressible form, the benchmark is approached in the low Mach number limit by setting $\mathrm{M}_0=0.1$, with $\gamma=1.4$ and $\mathrm{Pr}=0.71$.\\
\begin{figure}[htbp]
 \centering
 \subfigure[$t/t_{\mathrm{ref}} = 0$]
 {\includegraphics[width=0.485\textwidth]{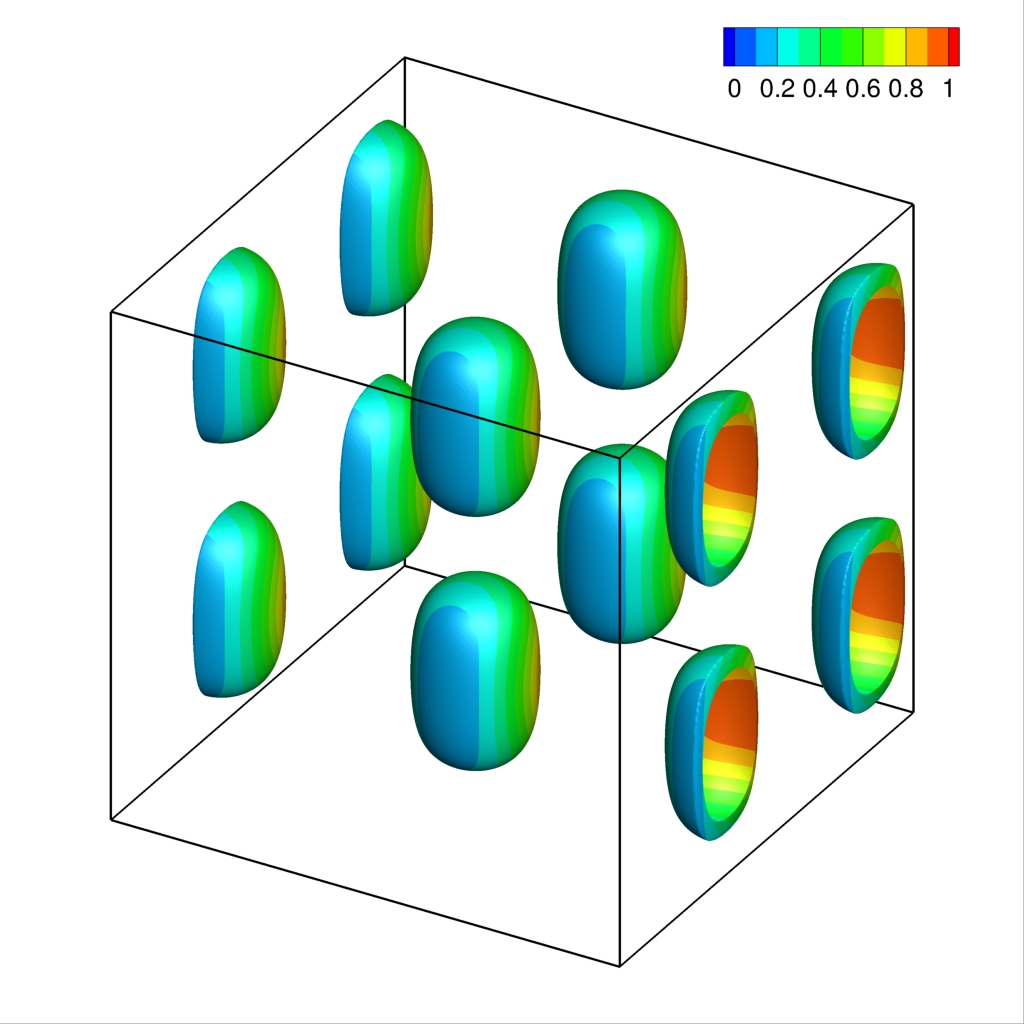}}
 \subfigure[$t/t_{\mathrm{ref}} = 20$]
 {\includegraphics[width=0.485\textwidth]{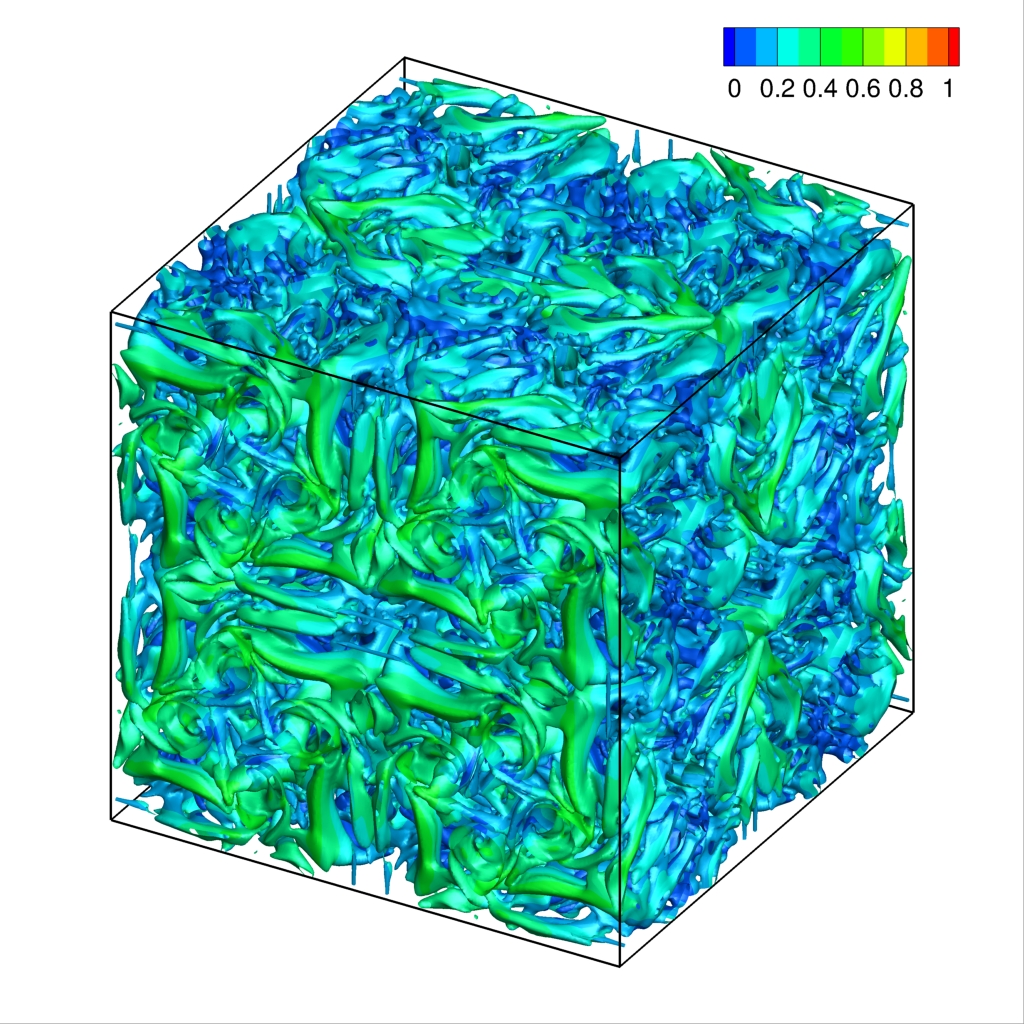}}
\caption{Taylor--Green Vortex at $\mathrm{Re} = 1600$, vortical structures visualized through Q criterion iso-surfaces and coloured by the dimensionless velocity magnitude. LES with standard LSQ filter, $256^3$ cells, NVIDIA GH200 based solution.}
\label{fig:tgv-q}
\end{figure}
%
%
The numerical setup follows the well established configuration  for the Taylor--Green vortex documented by Zirwes et al.,~\cite{Zirwes2023}. 
Thus, the computations were performed with the \texttt{rhoPimpleFoam} solver using the
 second--order backward scheme for time--integration.
As regards the spatial discretization, convective terms were approximated 
using a cubic interpolation scheme,~\cite{Zirwes2023}, whereas diffusive 
terms were treated with a cubic orthogonal scheme,~\cite{Zirwes2023}. 
Since the meshes are fully orthogonal, non-orthogonal corrections were unnecessary; similarly, the periodic computational domain and the absence of far--field boundaries obviated the need for any non--reflecting treatment.
The pressure--velocity coupling was handled through the PIMPLE algorithm, 
using two outer correctors and six pressure correctors per time step, with 
the momentum predictor enabled. Regarding the solution of the 
linear systems, pressure and density were treated using a preconditioned 
conjugate--gradient method with diagonal incomplete--Cholesky preconditioning, 
whereas velocity and energy were solved using PBiCGStab with diagonal 
incomplete-LU preconditioning. For all linear solvers, an absolute tolerance 
of $10^{-13}$ and zero relative tolerance were employed. A representative solution is shown in Fig.~\ref{fig:tgv-q}.\\
%
The global evolution of the TGV is characterized in terms of kinetic energy and kinetic energy dissipation rate. 
A distinction is made between the unmodelled and LES computations, since in the latter only the
resolved part of the flow field is directly available.
For the unmodelled computations, the kinetic energy is defined as
\begin{equation}
K(t)
=
\frac{1}{\rho_0 |\Omega|}
\int_{\Omega}
\left(
\frac{1}{2}\rho
u_i u_i
\right)
\,d\Omega~.
\label{eq:tgv-K}
\end{equation}
Following de Wiart et al.~\cite{deWiart:2014}, the corresponding kinetic energy dissipation is evaluated as
\begin{equation}
\varepsilon
=
\varepsilon_1+\varepsilon_2~,
\label{eq:tgv-eps}
\end{equation}
where
\begin{equation}
\varepsilon_1
=
2\frac{\mu_0}{\rho_0}
\frac{1}{|\Omega|}
\int_{\Omega}
\mathbf{S}^{d}:\mathbf{S}^{d}
\,d\Omega~,
\label{eq:tgv-eps1}
\end{equation}
and
\begin{equation}
\varepsilon_2
=
-\frac{1}{\rho_0|\Omega|}
\int_{\Omega}
p \frac{\partial{u}_i}{\partial x_i}
\,d\Omega~.
\label{eq:tgv-eps2}
\end{equation}
Here, $\mathbf{S}^{d}$ represents the deviatoric 
part of the strain-rate tensor.
As regards LES computations, the same quantities are 
evaluated from the resolved flow variables,~\emph{i.e.} 
by replacing $\rho$, $\mathbf{u}$ and $p$   with their corresponding resolved filtered quantities.
In this case, $\varepsilon$ represents the kinetic energy
dissipation associated with the resolved flow field and does not include the SGS energy transfer contribution.
For compactness, the same notation is retained in the figures for both unmodelled and LES computations.\\
%
%
%
A first set of computations was performed on uniformly spaced Cartesian grids, with the face area vectors aligned with the Cartesian coordinate directions. The time step was adjusted dynamically by imposing a maximum Courant number $\mathrm{Co}_{\max}=0.2$. To assess the suitability of this numerical setup for the TGV, a CPU based grid refinement study was carried out using $128^3$, $256^3$, and $512^3$ cells, with the finest computation run on 4096 CPU cores. For each computational grid, both unmodelled simulations and LES using the standard LSQ filter--width formulation were considered.\\
Fig.~\ref{fig:tgv1} compares the time--history of kinetic energy
and its dissipation rate with the spectral reference
solution,~\cite{deWiart:2014}. The kinetic energy decay, Fig.~\ref{fig:tgv1}(a),
is only weakly affected by grid resolution and remains in close
agreement with the reference solution over the complete evolution. This is why the kinetic energy plots are omitted in the following for compactness.
By contrast, the dissipation rate, shown in Fig.~\ref{fig:tgv1}(b), provides a more
discriminating assessment of the numerical resolution. It is readily observed as the unmodelled solutions exhibit a clear convergence toward the spectral reference as
the grid is refined, with the $512^3$ computation closely reproducing
both the magnitude and timing of the dissipation peak. The largest
discrepancies are observed on the coarser $128^3$ grid, particularly
around the peak dissipation stage, where the smallest dynamically
generated scales are only marginally resolved.
The unmodelled computations show that the adopted spatial and temporal discretization, together with the iterative solvers settings described above, is able to reproduce the expected high resolution evolution of the TGV. On the other hand, the LES solutions obtained with the standard LSQ filter-width exhibit the expected additional SGS dissipation, while progressively approaching the unmodelled high-resolution behaviour as the mesh is refined. Therefore, these results support the suitability of the numerical setup for the present TGV simulations and provide a consistent reference for assessing the behaviour of the LES filter-width formulations.\\
%
%
%
\begin{figure}[htbp]
 \centering
 \subfigure[Time history of the kinetic energy.]
 {\includegraphics[width=0.485\textwidth]{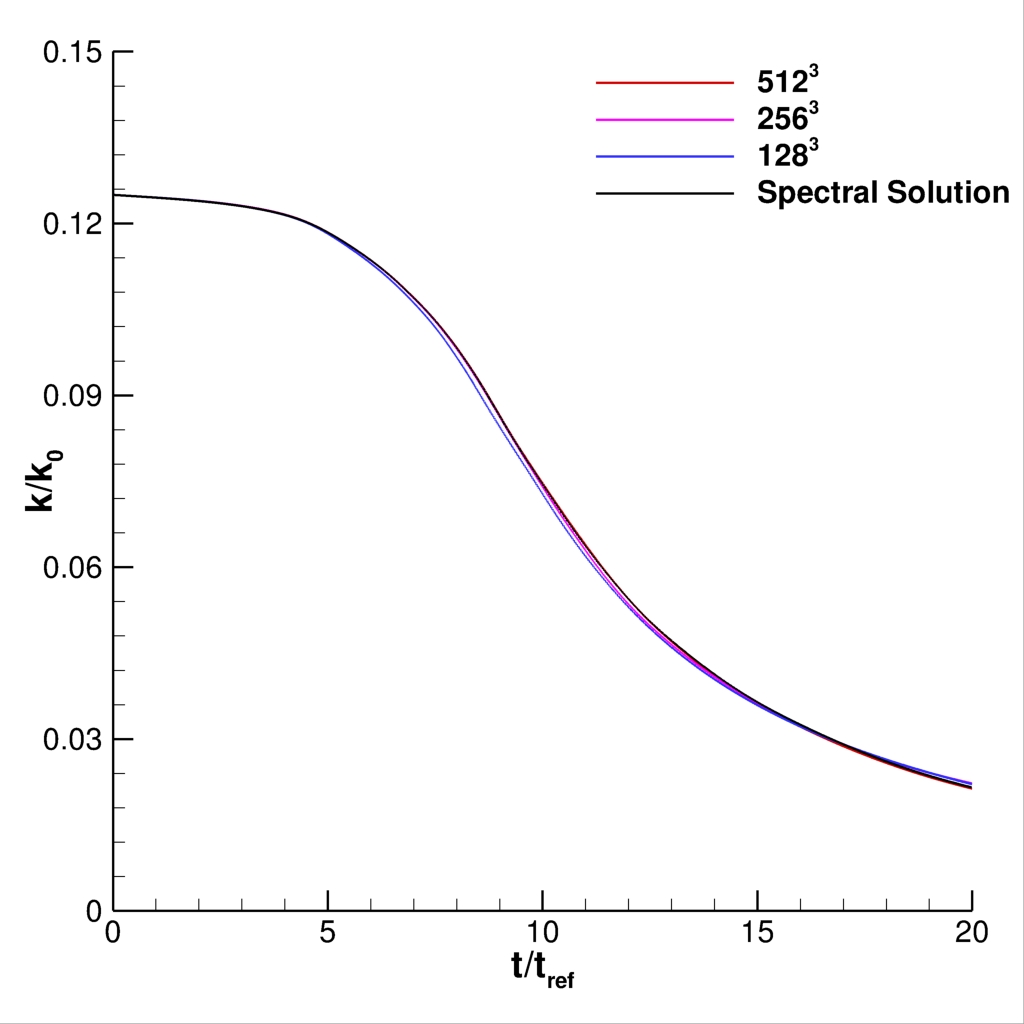}}
 \subfigure[Time history of the kinetic energy dissipation rate.]
 {\includegraphics[width=0.485\textwidth]{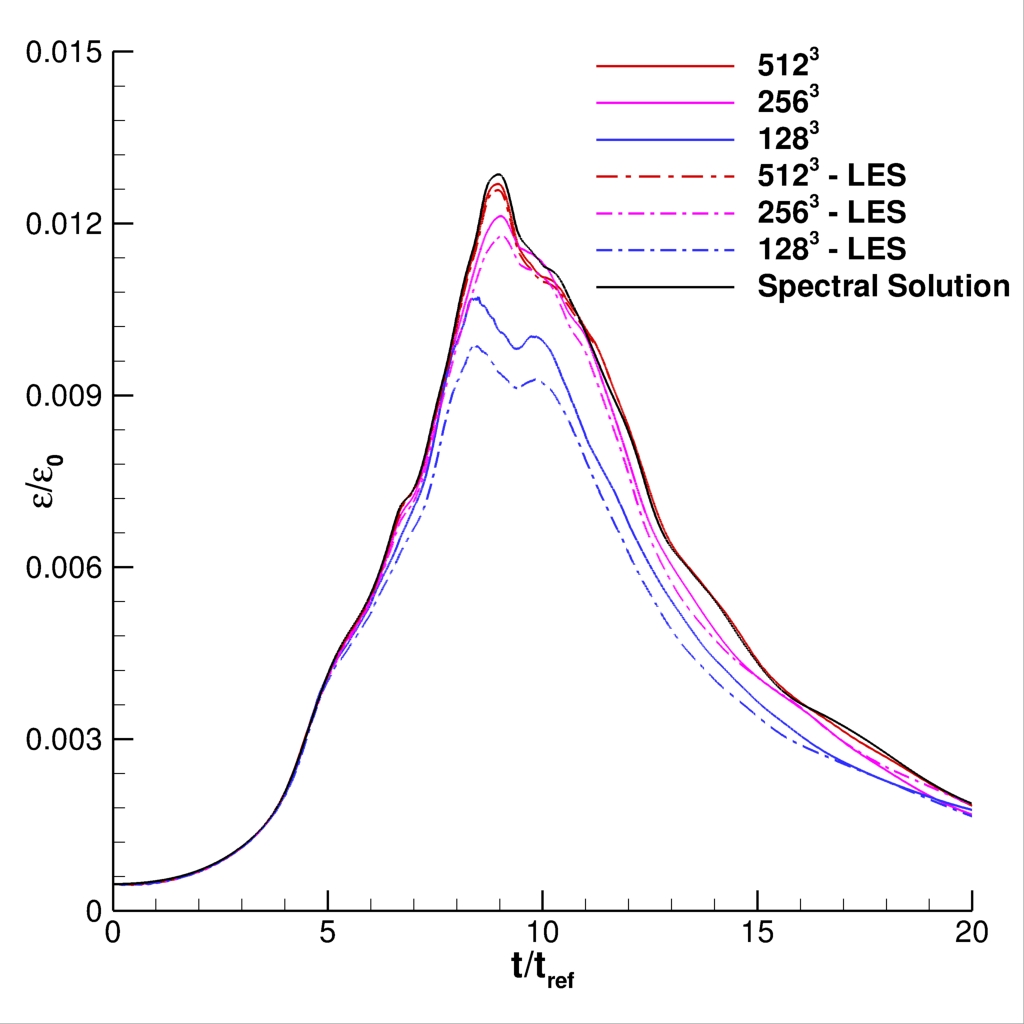}}
\caption{Taylor--Green Vortex at $\mathrm{Re} = 1600$, CPU solutions. LES are performed using standard LSQ filter width.}
\label{fig:tgv1}
\end{figure}
\begin{figure}[htbp]
 \centering
 \subfigure[Solution on $256^3$ cells using different architectures.]
 {\includegraphics[width=0.485\textwidth]{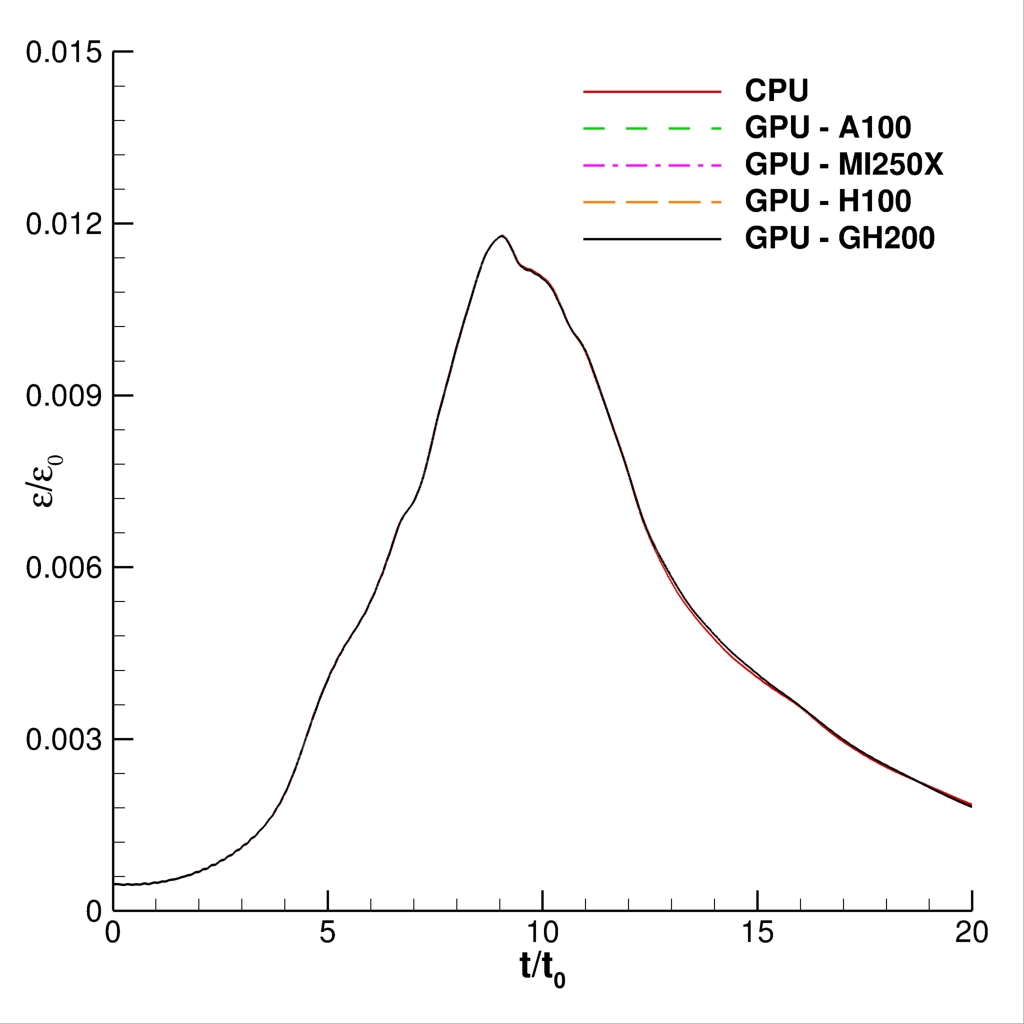}}
 \subfigure[NVIDIA A100 GPUs.]
 {\includegraphics[width=0.485\textwidth]{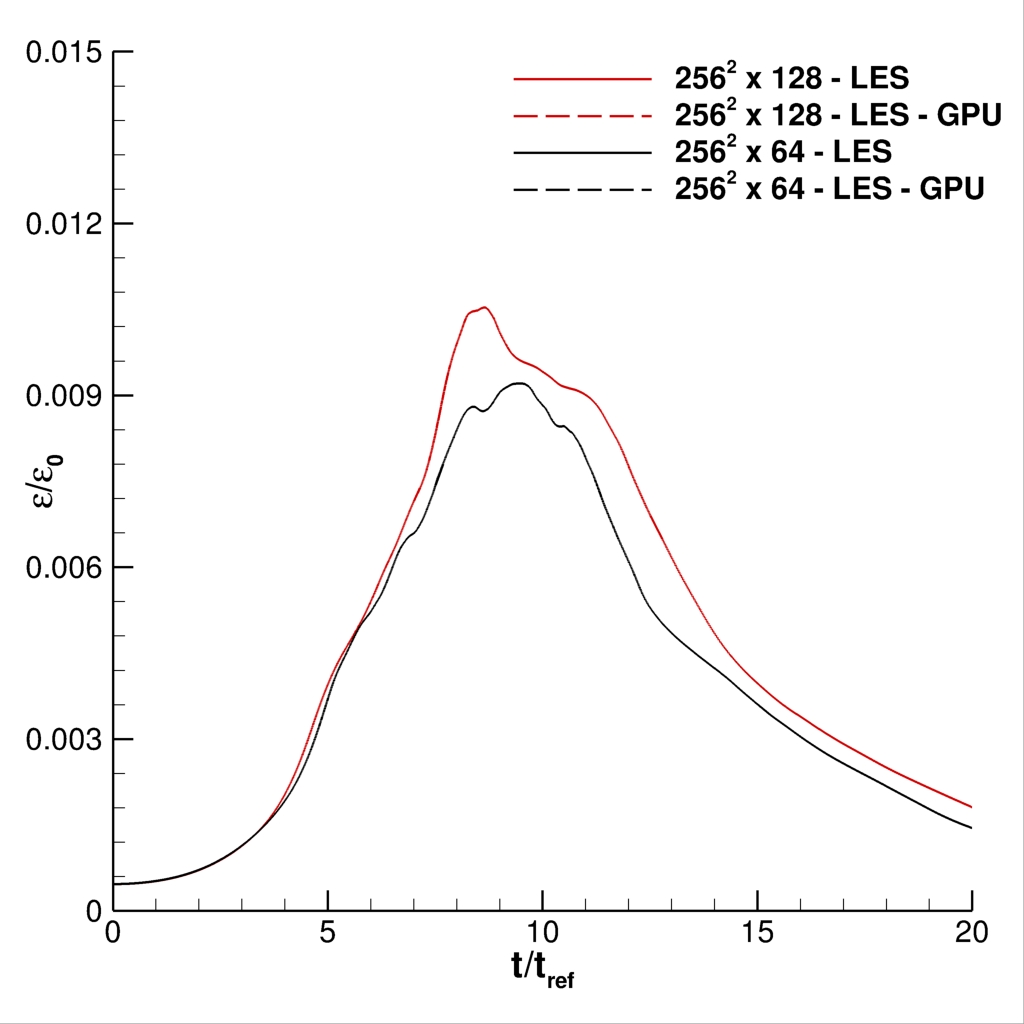}}
\caption{Taylor--Green Vortex at $\mathrm{Re} = 1600$. Time history of kinetic energy dissipation rate, standard LSQ filter--width. GPU solutions are obtained with face--based accumulations.}
\label{fig:tgv2}
\end{figure}
%
The sensitivity of the standard LSQ formulation to the computing architecture is first examined on the isotropic $256^3$ grid, as shown in Fig.~\ref{fig:tgv2}(a). It is clear that the kinetic energy dissipation time histories obtained on the CPU and on all the considered GPU architectures are virtually indistinguishable over the entire flow evolution. 
It is very important to note that for these isotropic meshes, only the two faces normal to the $x_i$-direction contribute to the corresponding directional metric sum in eq.~\ref{eq:directional-scale-accumulator}:
\[
\mathcal{A}_{i}
=
|S_{x_i^-}| + |S_{x_i^+}| \,~.
\]
%
Since the sum involves only two non--zero terms, its value is independent of the accumulation order. Consequently, the non--associative perturbation mechanism identified for the directional metric assembly is absent on the Cartesian TGV grids.
Note that this observation does not imply bitwise reproducibility of the complete CFD solution. Indeed, other architecture dependent floating--point differences may arise from the remaining numerical operations, including flux evaluation, iterative linear solvers, reductions, and time integration. 
Therefore, the small residual differences visible in Fig.~\ref{fig:tgv2}(a) are not unexpected. More importantly, once the specific non--associative metric assembly pathway is removed by the Cartesian mesh structure, no comparable modification of the resolved flow develops.
%
Results for the $256^2\times128$ and $256^2\times64$ grids are shown in Fig.~\ref{fig:tgv2}(b). Reducing the spanwise resolution modifies the kinetic energy dissipation time--history, particularly around the dissipation peak, while the CPU and NVIDIA A100 solutions remain closely matched for each grid. Notably,  as expected, spanwise coarsening changes the magnitude and anisotropy of the directional length scales but does not alter the two contribution structure of their Cartesian accumulation.

%
\begin{figure}[htbp]
 \centering
 \subfigure[CPU solutions.]
 {\includegraphics[width=0.485\textwidth]{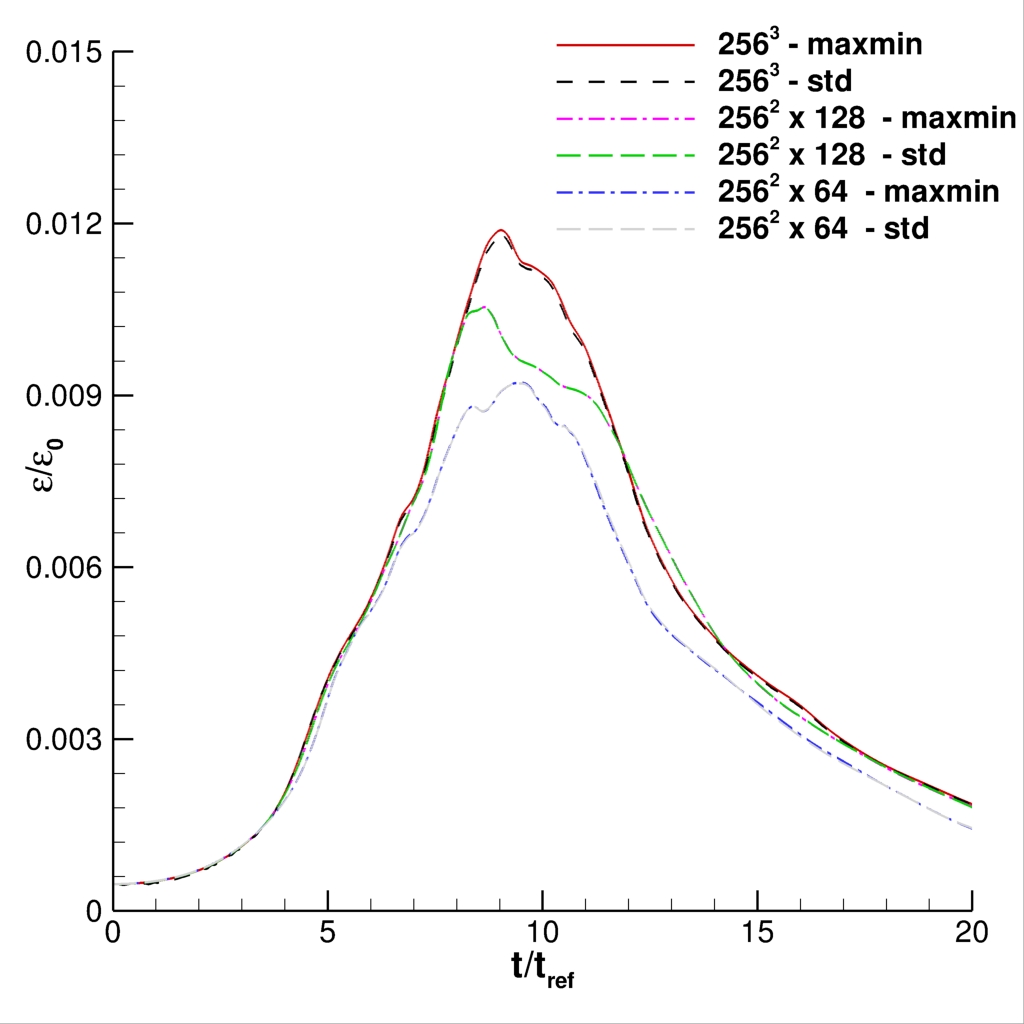}}
 \subfigure[CPU vs NVIDIA A100 solutions. LSQ max-min filter width.]
 {\includegraphics[width=0.485\textwidth]{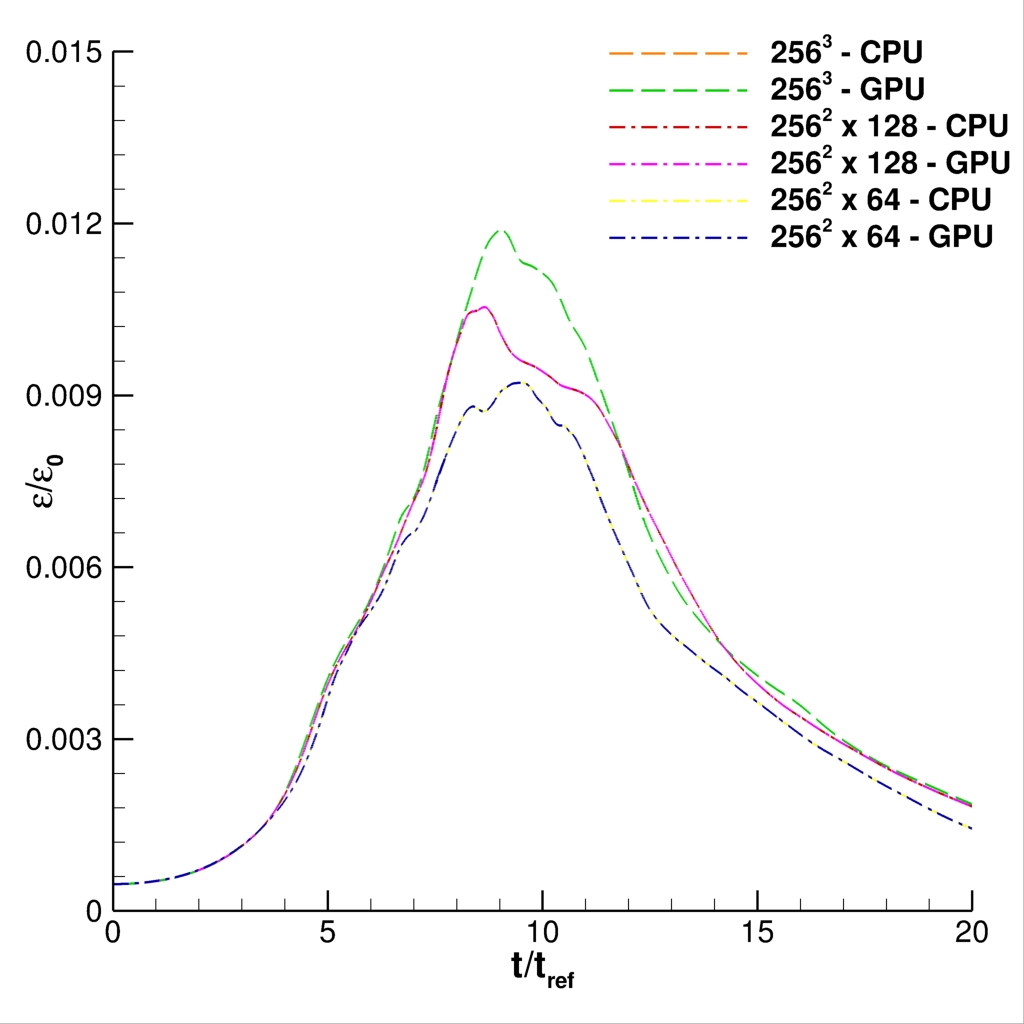}}
\caption{Taylor--Green Vortex at $\mathrm{Re} = 1600$. Comparison of  standard LSQ and LSQ--maxmin filter--width. Face--based accumulations.}
\label{fig:tgv-maxmin}
\end{figure}
The analysis is then extended to the proposed LSQ--maxmin formulation, with particular attention to CPU--GPU consistency. Specifically, Fig.~\ref{fig:tgv-maxmin}(a) compares the kinetic energy dissipation time--histories obtained on the CPU using the standard LSQ and LSQ--maxmin filter widths.
On the isotropic $256^3$ grid, the standard and LSQ--maxmin formulations produce almost identical solutions. A similarly close agreement is observed on the anisotropic $256^2\times128$ and $256^2\times64$ grids, where $\varepsilon$ values are very close at each resolution.
As expected, reducing the spanwise resolution modifies the kinetic energy dissipation time history, particularly around the peak dissipation stage. However, the response to increasing grid anisotropy is nearly identical for the two filter-width formulations. For this reason, the proposed metric compression can be considered to preserve, at the solution level, the response to grid anisotropy observed with the standard LSQ formulation.
Finally, Fig.~\ref{fig:tgv-maxmin}(b) further compares the LSQ--maxmin results
obtained on the CPU and on the NVIDIA A100 accelerator including also an isotropic grid.  The corresponding kinetic energy dissipation histories remain closely
matched over the complete flow evolution. The agreement is preserved
also for the most anisotropic $256^2\times64$ configuration, confirming
that the proposed formulation maintains CPU--GPU consistency on isotropic grids without
altering the characteristic response of the filter to spanwise grid
coarsening. Finally, it is important to note that for the LSQ--maxmin formulation, $C_\Delta=1.3$ is used throughout these TGV computations. This value is not calibrated on the TGV, but corresponds to the setting for which CPU and NVIDIA A100 simulations of the cylinder flow provide essentially identical predictions.\\
%
%
An additional set of TGV computations was performed on a smoothly distorted mesh constructed as follows.
Starting from the uniform Cartesian grid, the nodal coordinates were perturbed according to
$x'=x+\delta x$, $y'=y+\delta y$, and $z'=z+\delta z$~; where
\begin{equation}
\begin{aligned}
\delta x
&=
\varepsilon_g \, \mathcal{B}
\left[
\sin(2 Y)+\alpha\sin(3 Z)
\right]L~, \\
\delta y
&=
\varepsilon_g \, \mathcal{B}
\left[
\sin(2 Z)+\alpha\sin(3 X)
\right]L~, \\
\delta z
&=
\varepsilon_g \, \mathcal{B}
\left[
\sin(2 X)+\alpha\sin(3 Y)
\right]L~,
\end{aligned}
\label{eq:TGV-pert-nodes}
\end{equation}
$\mathcal{B} = \sin(X)\sin(Y)\sin(Z)~$, and $\alpha=0.35$~.
In the above equations the normalized coordinates are defined as follows:
\[
X=\frac{x+\pi L}{2 L}~,
\qquad
Y=\frac{y+\pi L}{2 L}~,
\qquad
Z=\frac{z+\pi L}{2 L}~.
\]
%
The parameter $\varepsilon_g$ controls the distortion amplitude. Note that the condition $\mathcal{B}=0$ at the domain boundaries ensures that the external geometry and the original number of cells are preserved as the interior mesh is smoothly deformed away from the initial isotropic setup.
The results reported below were obtained using $\varepsilon_g=5\cdot10^{-3}$ on the $256^3$ mesh. 
The value for $\varepsilon_g$ was selected since it is already sufficient to activate the full directional metric-assembly pathway: approximately $98\%$ of the cells exhibit six non-zero face contributions in each Cartesian direction.\\
It is clear from Fig.~\ref{fig:tgv-dis} that when the distorted grid is used, the global evolution of the Taylor--Green vortex remains essentially unchanged with respect to the orthogonal--grid solution,~\emph{i.e.} $\varepsilon_g=0$.
This result indicates that the presence of a pervasive non--associative accumulation pathway is not sufficient to produce a macroscopically observable modification of the resolved dynamics. 
Indeed, in the cylinder flow, the strongly concentrated kinematic weighting within the separated shear layers allows relative metric perturbations to be transmitted with near--unit gain. 
By contrast, in the TGV, the perturbation is distributed among the directional contributions, reducing the efficiency of the transmission pathway.
\begin{figure}[htbp]
 \centering
 \subfigure[Time--history of kinetic energy.]
 {\includegraphics[width=0.485\textwidth]{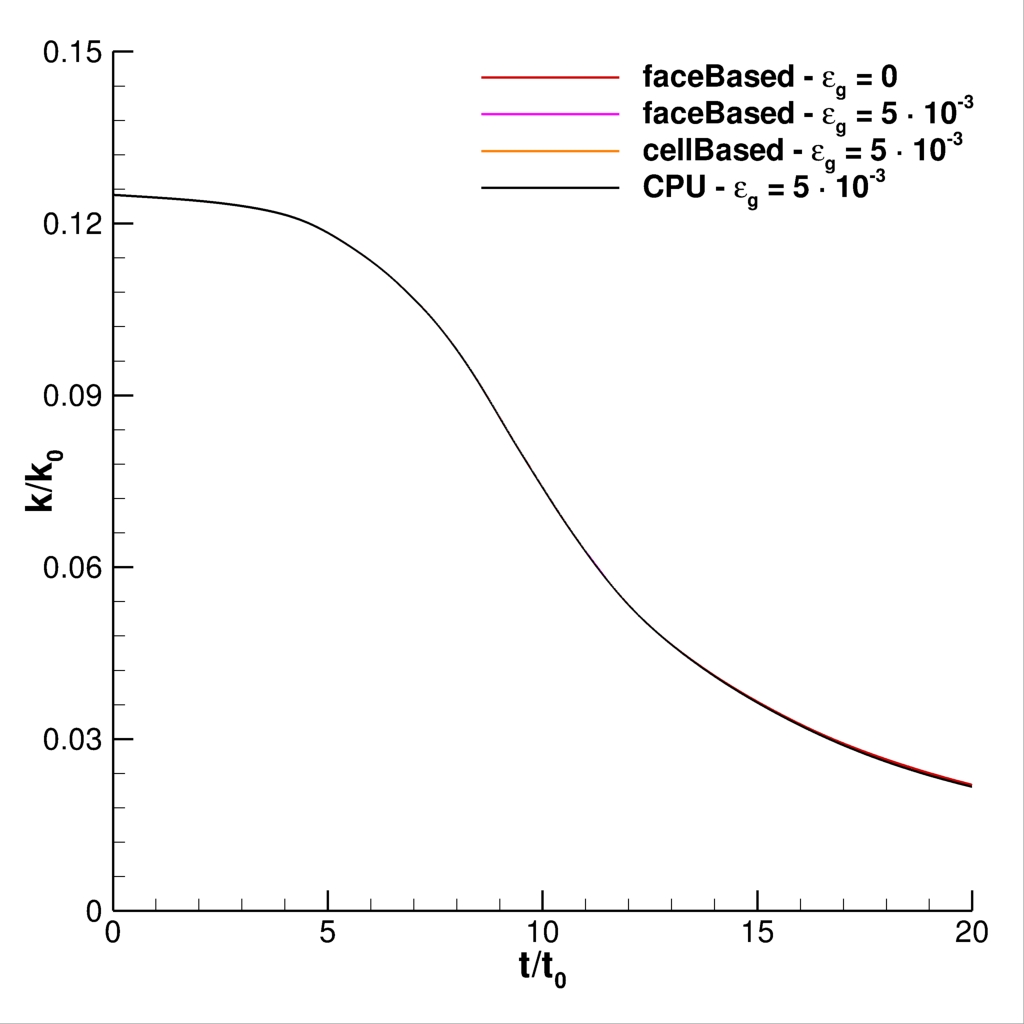}}
 \subfigure[Time--history of kinetic energy dissipation rate]
 {\includegraphics[width=0.485\textwidth]{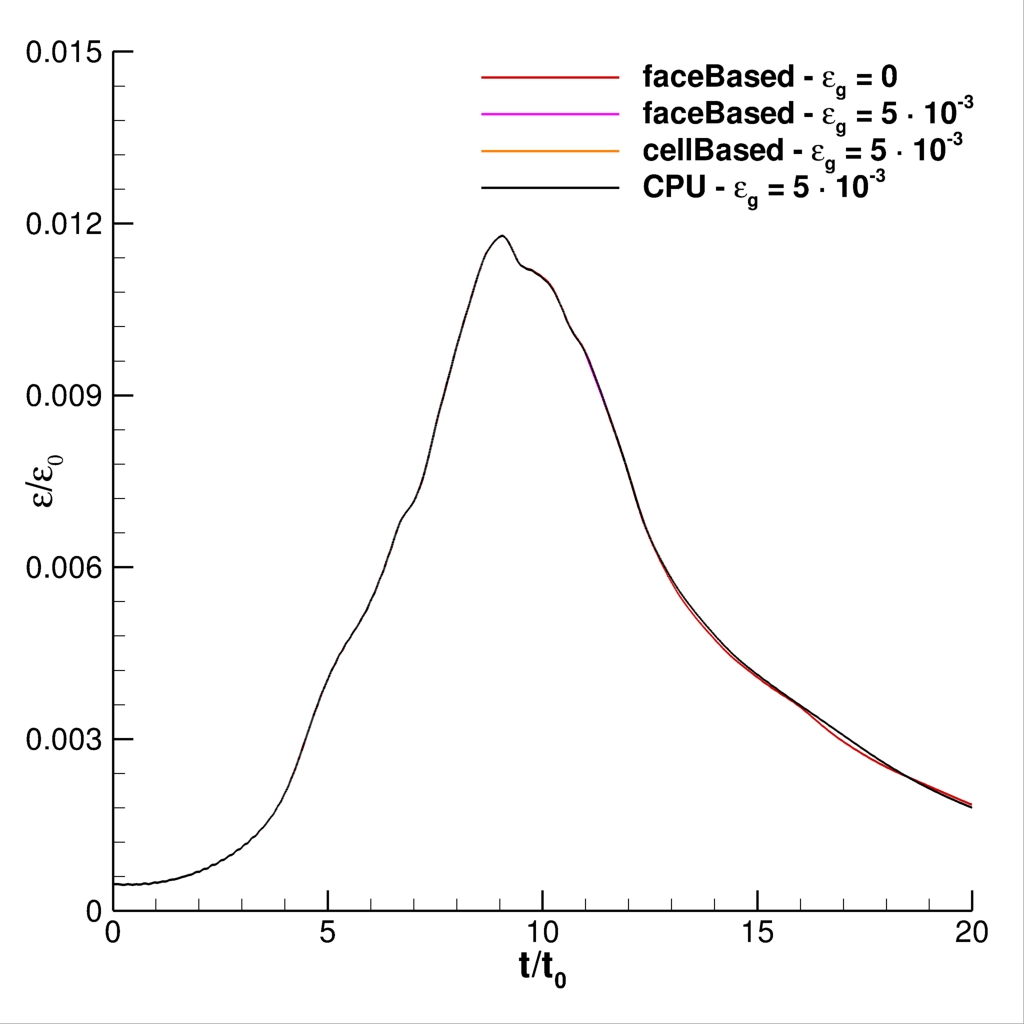}}
\caption{Taylor--Green Vortex at $\mathrm{Re} = 1600$. Results obtained on a distorted grid using NVIDIA A100 accelerators.  LSQ filter, face--based accumulations. }
\label{fig:tgv-dis}
\end{figure}
\subsection{Impact on other flow configurations} 
The numerical results suggest that three conditions must occur together for floating-point perturbations to become dynamically relevant: they must be introduced during the assembly of the directional metrics, efficiently transmitted through the filter-width operator, and act on a flow sufficiently receptive to the resulting SGS variations.\\
The cylinder and TGV cases help separate these effects. In the cylinder flow, the strongly concentrated LSQ weighting within the separated shear layers provides an efficient transmission pathway, while the transitional near-wake dynamics is sufficiently sensitive for the resulting SGS variations to affect the resolved flow. In the distorted grid TGV, metric-assembly perturbations are also present, but their effect is distributed among the directional contributions and no comparable macroscopic response develops. The orthogonal grid TGV provides the complementary limiting case, where the relevant non-associative accumulation pathway is absent.
These observations suggest that separated shear layers, free shear layers, and mixing layers are natural candidates for this mechanism, since the LSQ weighting may become strongly concentrated along selected directions. If such regions are also transitional, small variations in SGS dissipation may affect shear-layer breakdown, transition onset, or wake development. However, flow receptivity alone is not sufficient. On perfectly orthogonal Cartesian grids, the specific metric-assembly pathway considered here is suppressed. 
%
Therefore, flow fields combining transitional or strongly directional shear regions with geometrically non--orthogonal meshes appear to be the most plausible candidates for exhibiting the mechanism identified in this work.

\section{Conclusions}
\label{sec:conclusions}
The present work investigated the sensitivity of flow-dependent LES
filter--widths to floating--point perturbations introduced during the
parallel evaluation of directional mesh metrics. 
The perturbation transmissivity is established as a relevant design property 
of flow-dependent LES operators for heterogeneous computing.
For the standard LSQ
formulation, the analysis showed that, for fixed resolved kinematics,
the filter width is logarithmically non-expansive but not strictly
contractive. Consequently, directional metric perturbations may be
transmitted with little attenuation when the LSQ weighting becomes
strongly concentrated along one direction.\\
The circular cylinder results at $\mathrm{Re}=3900$ provide strong evidence for
this mechanism. 
The standard LSQ formulation transmits architecture dependent perturbations 
from the face--based metric assembly to the near--wake transition, whereas a 
deterministic cell--based metric evaluation, used here solely as a diagnostic control, 
suppresses this perturbation source and restores the physically consistent wake state.
%
Furthermore, a small diagnostic perturbation of the
kinematic weighting leads to the same qualitative recovery while leaving
the original face-based metric assembly unchanged. These results indicate
that the observed sensitivity arises from the interaction among metric
perturbations, the local LSQ kinematic weighting, and the strongly
transition sensitive separated shear-layer.
In order to control this mechanism, a max--min compression of the directional
mesh scales was introduced together with a bounded kinematic modulation,
leading to the factorized formulation
$\Delta_{\mathrm{LSQ,c}}=\Delta\,\mathrm{R}(\mathbf{G})$.
Here, the compressed metric scale $\Delta$ depends on the parameter
$C_\Delta$, which controls the activation of the minimum scale branch.
Note that the proposed LSQ--maxmin filter recovers the physically consistent
near-wake state on all the CPU and GPU architectures considered while
retaining the original face--based metric assembly. Although the
best performing value of $C_\Delta$ is not universal, all architectures
recover the physical wake state within the narrow interval
$1.1\leq C_\Delta\leq1.3$.
Moreover, the architecture specific values identified from the near--wake streamwise velocity profiles yield tightly clustered integral quantities, $\langle C_D\rangle\simeq1.00$, $C_{L,\mathrm{rms}}\simeq0.157$, and $\mathrm{St}\simeq0.214$, confirming a consistent macroscopic flow state.\\
The Taylor--Green vortex provides a complementary control case. For the
Cartesian meshes considered here, each directional metric is assembled
from only the two faces normal to the corresponding coordinate direction.
Therefore,  the order dependent accumulation mechanism identified on the cylinder
mesh is structurally absent. Hence, only small residual
differences between computing architectures are observed, while the
proposed LSQ--maxmin formulation preserves, at the solution level, the
response of the standard filter on both isotropic and spanwise--coarsened
grids. The distorted grid TGV results further show that multi--face directional metric accumulation alone is insufficient to produce an architecture dependent response, emphasizing the role of flow receptivity.\\
Overall, the obtained results indicate that relevant
architecture dependence requires the simultaneous occurrence of three
elements: a numerical source of floating--point perturbations, an
efficient transmission pathway through the filter--width operator, and
a receptive resolved flow field.

\section{Acknowledgements}
We acknowledge the EuroHPC Joint Undertaking for awarding this project access to the
EuroHPC supercomputers LEONARDO hosted by CINECA (Italy), JUPITER hosted by
JSC (Germany), MareNostrum5 hosted by BSC (Spain) and
LUMI hosted by CSC (Finland) through EuroHPC Development Access calls.\\
Part of the computing resources and the related technical support used for this work have been provided by CRESCO/ENEAGRID High Performance Computing infrastructure and its staff,~\cite{iannone2019cresco}. CRESCO/ENEAGRID High Performance Computing infrastructure is funded by ENEA, the Italian
National Agency for New Technologies, Energy and Sustainable Economic Development and by Italian
and European research programmes, see http://www.cresco.enea.it/english for information.
\appendix
\renewcommand{\thetheorem}{\Alph{section}.\arabic{theorem}}
\section{Perturbation stability of the max--min metric mapping}
\label{app:metric-compression}
This appendix provides the mathematical results underlying the
metric-perturbation analysis presented in Sec.~\ref{sec:sub-maxmim}.
In particular, it establishes the structural properties and logarithmic
stability of the max--min mapping used to compress the directional scales
into the scalar metric length $\Delta$.
%
\label{app:admissible-compression}

\begin{definition}[Logarithmically stable scalar mapping]
\label{def:log-stable-mapping}
A function $\mathcal{C}:(0,+\infty)\rightarrow(0,+\infty)$
is called a logarithmically stable scalar mapping if it is non decreasing and,
for every $s_0>0$, there exists $L_{\mathcal{C}}(s_0)<+\infty$ such that
\begin{equation}
\left|
\mathcal{C}(s_2)-\mathcal{C}(s_1)
\right|
\leq
L_{\mathcal{C}}(s_0)
\left|
s_2-s_1
\right|~,
\qquad
s_1,s_2\geq s_0~.
\label{eq:compression-lipschitz}
\end{equation}
Moreover, its logarithmic sensitivity is bounded, i.e. there exists
$\kappa_{\mathcal{C}}<+\infty$ such that
\begin{equation}
\left|
\log
\frac{\mathcal{C}(s_2)}
     {\mathcal{C}(s_1)}
\right|
\leq
\kappa_{\mathcal{C}}
\left|
\log
\frac{s_2}{s_1}
\right|,
\qquad
s_1,s_2>0.
\label{eq:compression-log}
\end{equation}
The mapping is called logarithmically non--expansive if
$\kappa_{\mathcal{C}}\leq 1$, and logarithmically contractive if
$\kappa_{\mathcal{C}}<1$.
\end{definition}

\label{app:lower-bound-composition}

\begin{proposition}[Perturbation stability under lower-bound composition]
\label{prop:lower-bound-composition}
Let $\mathcal{C}:(0,+\infty)\rightarrow(0,+\infty)$ be a
logarithmically stable scalar compression and define
\begin{equation}
\Delta
=
\mathcal{C}\!\left(\max\{\Psi,B\}\right)~,
\qquad
\Psi,B>0~.
\label{eq:lower-bound-composition}
\end{equation}
For perturbed quantities $\widehat{\Psi},\widehat{B}>0$, let
\begin{equation}
\widehat{\Delta}
=
\mathcal{C}\!\left(\max\{\widehat{\Psi},\widehat{B}\}\right)~.
\end{equation}
Then the composed mapping is non-decreasing in both arguments and preserves the absolute and logarithmic perturbation bounds of 
$\mathcal{C}$. In particular, if
\begin{equation}
\max\{\Psi,B\}\geq s_0~,
\qquad
\max\{\widehat{\Psi},\widehat{B}\}\geq s_0~,
\qquad
s_0>0~,
\end{equation}
then
\begin{equation}
\left|\widehat{\Delta}-\Delta\right|
\leq
L_{\mathcal{C}}(s_0)
\max\left\{
\left|\widehat{\Psi}-\Psi\right|~,
\left|\widehat{B}-B\right|
\right\},
\label{eq:composed-absolute-bound}
\end{equation}
and
\begin{equation}
\left|
\log
\frac{\widehat{\Delta}}{\Delta}
\right|
\leq
\kappa_{\mathcal{C}}
\max\left\{
\left|
\log
\frac{\widehat{\Psi}}{\Psi}
\right|,
\left|
\log
\frac{\widehat{B}}{B}
\right|
\right\}~.
\label{eq:composed-log-bound}
\end{equation}
\end{proposition}

\begin{proof}
The monotonicity of the composed mapping follows directly from the
monotonicity of both \(\mathcal{C}\) and the scalar maximum. Moreover, the
maximum operator is non-expansive, since
\[
\left|
\max\{\widehat{\Psi},\widehat{B}\}
-
\max\{\Psi,B\}
\right|
\leq
\max
\left\{
|\widehat{\Psi}-\Psi|,
|\widehat{B}-B|
\right\}~.
\]
Combining this inequality with the Lipschitz continuity of
\(\mathcal{C}\) gives the absolute perturbation bound. For the logarithmic
estimate, we use the identity
\[
\log\max\{\Psi,B\}
=
\max\{\log\Psi,\log B\},
\]
valid for positive \(\Psi\) and \(B\). The non--expansiveness of the maximum
in logarithmic variables, together with the logarithmic sensitivity bound of
\(\mathcal{C}\), then yields the stated result.
\end{proof}
%
\label{app:power-law-compression}

\begin{proposition}[Power-law mappings]
\label{prop:power-law}
Let $\mathcal{C}_{\alpha}(s)=s^{\alpha}~~,0<\alpha\leq 1~$. 
Then $\mathcal{C}_{\alpha}$ is an admissible metric-compression mapping with
\begin{equation}
L_{\mathcal{C}_{\alpha}}(s_0)
=
\alpha s_0^{\alpha-1}~,
\qquad
\kappa_{\mathcal{C}_{\alpha}}
=
\alpha~.
\label{eq:power-law-bounds}
\end{equation}
Hence, $\mathcal{C}_{\alpha}$ is logarithmically contractive for
$0<\alpha<1$.
\end{proposition}
\begin{proof}
For \(0<\alpha\leq1\), the mapping \(\mathcal{C}_{\alpha}\) is positive and
non-decreasing. Since $\mathcal{C}_{\alpha}'(s) = \alpha s^{\alpha-1}$ is non increasing, its maximum over \([s_0,+\infty)\) is attained at
\(s=s_0\). Hence,
\[
L_{\mathcal{C}_{\alpha}}(s_0)
=
\sup_{s\geq s_0}
\left|
\mathcal{C}_{\alpha}'(s)
\right|
=
\alpha s_0^{\alpha-1}~.
\]
Moreover, $ \log \mathcal{C}_{\alpha}(s) = \alpha\log s~, $ which gives
\[
\left|
\log
\frac{\mathcal{C}_{\alpha}(s_2)}
     {\mathcal{C}_{\alpha}(s_1)}
\right|
=
\alpha
\left|
\log
\frac{s_2}{s_1}
\right|~.
\]
Therefore,
\(\kappa_{\mathcal{C}_{\alpha}}=\alpha\), and the mapping is logarithmically
contractive for \(0<\alpha<1\)~.
\end{proof}
%
\subsection{Three-dimensional geometric-mean construction}
\label{app:geometric-construction}
To specialize the previous results to the max--min metric mapping adopted in the
present work, we introduce the aggregate quantities
\begin{equation}
\Psi
=
\Delta_x\Delta_y\Delta_z~,
\qquad
B
=
\left(C_{\Delta}\Delta_{\min}\right)^3~,
\qquad
\Delta_{\min}
=
\min\{\Delta_x,\Delta_y,\Delta_z\}~.
\label{eq:psi-b-definition}
\end{equation}
The max--min construction can then be obtained by choosing the power--law mapping $\mathcal{C}(s)=s^{1/3}$
and applying it to $Q=\max\{\Psi,B\}$~.
This yields
\begin{equation}
\begin{aligned}
\Delta
&=
\left[
\max\left\{
\Delta_x\Delta_y\Delta_z,
\left(C_{\Delta}\Delta_{\min}\right)^3
\right\}
\right]^{1/3}
\\
&=
\max\left\{
\left(\Delta_x\Delta_y\Delta_z\right)^{1/3},
C_{\Delta}\Delta_{\min}
\right\}.
\end{aligned}
\label{eq:three-dimensional-map}
\end{equation}
The first branch provides a symmetric scalar aggregation of the directional
length scales through their geometric mean, whereas the second branch introduces
an explicit dependence on the smallest directional scale by imposing the lower bound $\Delta \geq C_{\Delta}\Delta_{\min}$~.
Accordingly, $C_{\Delta}$ controls the activation of the minimum-scale branch
and therefore the location of the branch-switching surface.
Proposition~\ref{prop:power-law} gives
\begin{equation}
L_{\mathcal{C}}(s_0)
=
\frac{1}{3}s_0^{-2/3}~,
\qquad
\kappa_{\mathcal{C}}
=
\frac{1}{3}~.
\label{eq:cube-root-bounds}
\end{equation}
Hence, for $Q,\widehat{Q}\geq s_0>0$~,
\begin{equation}
\left|\widehat{\Delta}-\Delta\right|
\leq
\frac{1}{3}s_0^{-2/3}
\max\left\{
\left|\widehat{\Psi}-\Psi\right|~,
\left|\widehat{B}-B\right|
\right\},
\label{eq:cube-root-absolute}
\end{equation}
while the corresponding logarithmic perturbation satisfies
\begin{equation}
\left|
\log
\frac{\widehat{\Delta}}{\Delta}
\right|
\leq
\frac{1}{3}
\max\left\{
\left|
\log
\frac{\widehat{\Psi}}{\Psi}
\right|,
\left|
\log
\frac{\widehat{B}}{B}
\right|
\right\}~.
\label{eq:cube-root-log}
\end{equation}
These estimates describe the response of the cube-root mapping to perturbations
of the aggregate, volume-like quantities $\Psi$ and $B$. They do not yet provide
a stability bound directly in terms of perturbations of the directional scales,
since both aggregate quantities depend on $\Delta_x$, $\Delta_y$, and $\Delta_z$.
The perturbation properties of the complete three-variable mapping are therefore
examined next.
%
\subsection{Stability with respect to directional-scale perturbations}
\label{app:directional-scale-stability}

Define
\begin{equation}
F(\Delta_x,\Delta_y,\Delta_z)
=
\max\left\{
\left(\Delta_x\Delta_y\Delta_z\right)^{1/3},
C_{\Delta}\min\{\Delta_x,\Delta_y,\Delta_z\}
\right\}.
\label{eq:complete-maxmin-map}
\end{equation}
The mapping $F$ is positive, symmetric with respect to permutations of the
coordinate directions, non-decreasing in each argument, and positively
homogeneous:
\begin{equation}
F\!\left(
\lambda\Delta_x,
\lambda\Delta_y,
\lambda\Delta_z
\right)
=
\lambda F(\Delta_x,\Delta_y,\Delta_z),
\qquad
\lambda>0.
\label{eq:positive-homogeneity}
\end{equation}

\begin{proposition}[Logarithmic non-expansiveness]
\label{prop:directional-nonexpansiveness}
For any positive directional scales $\Delta_i$ and their perturbed
counterparts $\widehat{\Delta}_i$, with $i\in\{x,y,z\}$,
\begin{equation}
\left|
\log
\frac{
F(\widehat{\Delta}_x,\widehat{\Delta}_y,\widehat{\Delta}_z)
}{
F(\Delta_x,\Delta_y,\Delta_z)
}
\right|
\leq
\max_{i\in\{x,y,z\}}
\left|
\log
\frac{\widehat{\Delta}_i}{\Delta_i}
\right|.
\label{eq:directional-nonexpansive-bound}
\end{equation}
Therefore, the complete max--min construction is non-expansive in the
componentwise logarithmic $\ell_{\infty}$ metric.
\end{proposition}

\begin{proof}
Introduce the logarithmic variables
\begin{equation}
 u_i=\log\Delta_i,
 \qquad
 \widehat{u}_i=\log\widehat{\Delta}_i,
\end{equation}
and define
\begin{equation}
A(\mathbf{u})
=
\frac{u_x+u_y+u_z}{3},
\qquad
M(\mathbf{u})
=
\log C_{\Delta}+\min\{u_x,u_y,u_z\}.
\end{equation}
Then
\begin{equation}
\log F(\Delta_x,\Delta_y,\Delta_z)
=
\max\{A(\mathbf{u}),M(\mathbf{u})\}.
\end{equation}
Both $A$ and $M$ are non-expansive in the $\ell_{\infty}$ norm:
\begin{equation}
\left|A(\widehat{\mathbf{u}})-A(\mathbf{u})\right|
\leq
\left\|\widehat{\mathbf{u}}-\mathbf{u}\right\|_{\infty},
\end{equation}
and
\begin{equation}
\left|M(\widehat{\mathbf{u}})-M(\mathbf{u})\right|
\leq
\left\|\widehat{\mathbf{u}}-\mathbf{u}\right\|_{\infty}.
\end{equation}
The non-expansiveness of the scalar maximum therefore gives
\begin{equation}
\left|
\log F(\widehat{\boldsymbol{\Delta}})
-
\log F(\boldsymbol{\Delta})
\right|
\leq
\left\|\widehat{\mathbf{u}}-\mathbf{u}\right\|_{\infty},
\end{equation}
which is equivalent to eq.~\ref{eq:directional-nonexpansive-bound}.
\end{proof}
Eq.~\ref{eq:directional-nonexpansive-bound} provides the global
stability estimate used in Sec.~\ref{sec:sub-maxmim}. The branchwise
interpretation of the mapping and its response to directional metric
perturbations are discussed there.

\bibliographystyle{elsarticle-num}
\bibliography{main}

\end{document}